\documentclass[12pt,onecolumn,draftcls]{IEEEtran}
\usepackage{epsfig}
\usepackage{amsmath}
\usepackage{amsfonts}
\usepackage{amssymb}
\usepackage{color} 
\usepackage{float}
\usepackage{url}
\usepackage{flushend}
\usepackage{hyperref}
\usepackage{verbatim}
\usepackage{setspace}
\usepackage{hyperref}
\usepackage{subfigure}

\newcommand{\E}{\mathbb{E}}

\newcommand{\Log}{\textrm{log}}

\newcommand{\mb}{\mathbf}
\newcommand{\bs}{\boldsymbol}

\newcommand{\Prob}{\textrm{Pr}}
\newtheorem{definition}{Definition}
\newtheorem{lemma}{Lemma}
\newtheorem{proposition}{Proposition}

\newcommand{\argmax}[1]{\underset{#1}{\operatorname{argmax}}} 
\newcommand{\argmin}[1]{\underset{#1}{\operatorname{argmin}}} 

\title{Precoding for Outage Probability Minimization on Block Fading Channels}
\author{Dieter~Duyck,~Joseph~J.~Boutros,~and~Marc~Moeneclaey
\thanks{\noindent  Dieter  Duyck  and  Marc Moeneclaey  are  with  the
Department  of Telecommunications  and  Information processing,  Ghent
University,   St-Pietersnieuwstraat    41,   B-9000   Gent,   Belgium,
\{dduyck,mm\}@telin.ugent.be.  Joseph J.  Boutros is  with  Texas A\&M
University  at Qatar,  PO  Box 23874  Doha, Qatar,  boutros@tamu.edu}
\thanks{\copyright 2011 IEEE. Personal use of this material is permitted. Permission from IEEE must be obtained for all other uses, in any current or future media, including reprinting/republishing this material for advertising or promotional purposes, creating new collective works, for resale or redistribution to servers or lists, or reuse of any copyrighted component of this work in other works.}}

\begin{document}
\maketitle

\begin{abstract} 
The outage probability limit is a fundamental and achievable lower bound on the word error rate of coded communication systems affected by fading. This limit is mainly determined by two parameters: the diversity order and the coding gain. With linear precoding, full diversity on a block fading channel can be achieved without error-correcting code. However, the effect of precoding on the coding gain is not well known, mainly due to the complicated expression of the outage probability. Using a geometric approach, this paper establishes simple upper bounds on the outage probability, the minimization of which yields to precoding matrices that achieve very good performance.
 For discrete alphabets, it is shown that the combination of constellation expansion and precoding is sufficient to closely approach the minimum possible outage achieved by an i.i.d. Gaussian input distribution, thus essentially maximizing the coding gain.
\end{abstract}

\section{Introduction}

In many applications, using techniques such as frequency-hopping (GSM, EDGE), time-interleaving (DVB-T), OFDMA (WiMax and LTE), H-ARQ with cross-packet coding \cite{Hausl2007a, duy2010dod}, and cooperative  communications \cite{sendonaris2003ucd, sendonaris2003ucd2, laneman2004cdw, duy2009ldg}, the channel can be modelled as a flat block fading (BF) channel \cite{biglieri1998fci}, where the fading gain is piecewise constant over the duration of a transmitted packet. Due to motion of the transmitter, receiver or objects between the transmitter and receiver, the fading gains vary from one packet to the next and are considered unknown at the transmitter side. The fraction of codewords where decoding fails to wipe out all errors is referred to as the average word error rate (WER). When displayed on a log-scale versus the average signal-to-noise ratio (SNR) in decibel, the high-SNR slope of the WER is called the diversity order\footnote{We assume fading gain distributions where the WER can be expressed as $g \textrm{SNR}^{-\delta}$ where $g$ and $\delta$ are constants; and $\delta$ is the diversity order.}. Since the diversity order determines how fast the error rate decreases with the SNR, it is then a key parameter of the communication system.

By appropriately linearly precoding multidimensional signal constellations before transmitting each of its components on a different block of the block fading channel (denoted as component interleaving), the Singleton bound does not limit the maximal coding rate $R_c$, achieving full-diversity, any more \cite{gresset2004olp, gresset2008stc, fab2007mcm}. In particular, full diversity can be achieved without error-correcting code. Therefore, it has been almost exclusively studied for uncoded schemes (see \cite{bayer2004nac, biglieri1998fci} and references therein, e.g. \cite{boutros1998ssd}), except for few papers investigating outer coded transmission schemes for MIMO (e.g. \cite{gresset2004olp, kraidy2004olp}) and recent work on pairing subchannels when channel state information at the transmitter (CSIT) is available \cite{mohammed2011pbp, mohammed2010mpw}. However, besides yielding full diversity with a higher rate, linear precoding can also improve the coding gain of coded systems. So the effect of linear precoding on the diversity order is well understood, but there is no sufficient insight in the effect of linear precoding on coding gain of coded systems.

Before studying the optimization of the WER for practical schemes with linear precoding, it is important to understand the influence of linear precoding on the performance limits of the communication channel. The  outage probability limit  is an achievable lower bound on the WER of coded  systems in the limit of large block length \cite{biglieri1998fci}, \cite{ozarow1994itc} and is given by \cite[Sec. 5.4]{tse2005fwc},
\begin{equation*}
	P_{\textrm{out}}(\gamma, P, R) = \Prob\big(I(\bs{\alpha}, \gamma, P) < R \big),
\end{equation*}
\noindent where $I\left(\bs{\alpha}, \gamma, P\right)$ is the instantaneous mutual information between transmitted symbols and received symbols as a function of the set of fading gains $\bs{\alpha}$ observed during the transmission of one codeword, the average SNR $\gamma$, and the precoding matrix $P$. By choosing a well designed precoding matrix $P$, the outage probability can be minimized. But only a brute force optimization can minimize the outage probability as a closed form expression of the outage probability is not available. Such an optimization is often intractable when the number of fading gains per codeword is larger than two and/or large constellations (e.g. 16 points) are used. A simple approach could be designing $P$ such that the mean of $I\left(\bs{\alpha}, \gamma, P\right)$ over the fading distribution is maximized, in the hope that the area under the left tail of its probability density function (pdf) would be minimized. However, the ergodic mutual information $\E_{\bs{\alpha}}[I\left(\bs{\alpha}, \gamma, P\right)]$ contains no information on the diversity order, and, due to the limited spectral efficiency of finite discrete input alphabets, for increasing $\gamma$ rapidly converges to a maximum that does not depend on the precoding matrix. Therefore, the maximization of the ergodic mutual information fails to provide the optimum precoding matrix \cite{duy2010rmf}. 

In this paper, we first study the effect of linear precoding of discrete input alphabets on the outage probability\footnote{Note that for i.i.d. Gaussian input distributions, it was found that a power allocation (corresponding to a multiplication of the symbol vector with a non-unitary matrix) over the different blocks minimized the outage probability \cite[Sec. 5.1]{telatar1999com}, \cite{jorswieck2007opi}.}. A unitary $B \times B$ precoding matrix has $B^2$ parameters that are free to choose (denoted as degrees of freedom), so that the optimization of the precoding matrix is multivariate. In a brute force optimization, Monte Carlo simulations are required to take into account the distribution of all fading gains when computing the outage probability. Here, the analysis uses a geometric approach that leads to upper bounds on the outage probability, which are easier to optimize, because it is no longer necessary to perform a Monte Carlo simulation based on the fading gain distribution. 
Therefore, with a minimal computational effort, the precoding matrix minimizing the upper bound on the outage probability limit can be determined. These results serve as a basis for the design of practical coded systems with linear precoding \cite{duy2010cmf}, which is discussed in the final section before conclusions.

\section{System model\label{sec_system_model}}

At the transmitter output, a packet is represented as a real-valued or complex column vector $\bs{\chi} = [\bs{\chi}(1)^T, \ldots, \bs{\chi}(B)^T]^T$ of dimension $N$, consisting of $B$ blocks that each contain $N/B$ symbols, where $(.)^T$ indicates transposition; the b-th block of the packet is $\bs{\chi}(b) = [\chi(b)_1, \ldots, \chi(b)_{\frac{N}{B}}]^T$ with $\mathbb{E}\left[|\chi(b)_n|^2\right]=1$. The channel is memoryless with additive white Gaussian noise and multiplicative real-valued fading. The fading coefficients are only known at the decoder side. The received signal vector $\bs{\mu}(b)$ corresponding to the transmitted block $\bs{\chi}(b)$ is
\begin{equation}
	\label{eq: system model}
	\bs{\mu}(b) = \alpha_b \bs{\chi}(b) + \bs{\omega}(b), ~~b=1, \ldots, B.
\end{equation}
The  fading   coefficients  $\{\alpha_b, b=1, \ldots, B\}$ are i.i.d. In the numerical results, we consider the fading coefficients to be Rayleigh distributed, with $\mathbb{E}[\alpha_b^2]=1$, but the analysis in this paper does not depend on the fading distribution. The noise vector $\bs{\omega}(b)$ consists of $N/B$ independent noise  samples which  are  complex Gaussian  distributed, $\omega(b)_n \sim \mathcal{CN}(0,2\sigma^2)$. The average signal-to-noise ratio is $\gamma=\frac{1}{2\sigma^2}$.  

The transmitted vector $\bs{\chi}$ is obtained from the information bits through a sequence of operations. Assuming a binary encoder with coding rate $R_c$, a packet of $K$ information bits is encoded into $K/R_c$ coded bits. The binary codeword is split into $K/(m R_c)$ strings each containing $m$ bits. In a standard coded communications system, the components of the transmitted vector $\bs{\chi}(b)$ are obtained by directly mapping each string of $m$ coded bits to one of $M = 2^m$ points belonging to a 1-dimensional real or complex space; the corresponding spectral efficiency $R$ in bits per channel use (bpcu) is given by $R=m R_c$. When using precoding combined with component interleaving, each string of $m$ bits is mapped to one of $M = 2^m$ points belonging to a B-dimensional real or complex space; the corresponding B-dimensional M-point constellation $\Omega_z$ is denoted $M$-$\mathcal{R}^B$ or $M$-$\mathcal{C}^B$, respectively. Denoting as $\mb{z}(n) = [z(n)_1, \ldots, z(n)_B]^T$ the B-dimensional vector that results from mapping the n-th string of $m$ coded bits, the linear precoding involves the computation 
\begin{equation}
\label{eq: precoding}
\mb{x}(n)  = P \mb{z}(n), ~ n=1, \ldots, \frac{N}{B}
\end{equation}
where $P$ is a non-singular precoding matrix of dimension $B \times B$. The precoder output vectors $\mb{x}(n) = [x(n)_1, \ldots, x(n)_B]$ belong to a B-dimensional M-point constellation $\Omega_x$ which results from a linear transformation (through $P$) of $\Omega_z$. Finally, component interleaving is applied: the n-th component of the b-th block of the transmitted vector $\bs{\chi}$ equals the b-th component of the n-th precoder output vector, i.e., $\chi(b)_n = x(n)_b$ (Fig. \ref{fig: multidimensional modulator part 1}). Hence, the $B$ components of $\mb{x}(n)$ experience independent fading when transmitted over the BF channel. Taking into account that $K$ information bits are transformed into a transmitted vector $\bs{\chi}$ with $N = K B/(m R_c)$ components, the overall spectral efficiency is $R=R_c \frac{m}{B}$ bpcu. Note that there are several ways to achieve a given spectral efficiency $R$. For example, $R=0.9$ bpcu for $B=2$ can be achieved by a coded communication system with precoding and component interleaving ($R= R_c \frac{m}{B}$) by choosing $m=3$ and $R_c=0.6$, whereas a standard communication system ($R=R_c m$) achieves $R=0.9$ bpcu for $m=1$ and $R_c=0.9$. Also other combinations of $m$ and $R_c$ are possible. Note that a standard communication system with $R=m' R_c$ can be viewed as a special case of a multidimensional modulation, where $\Omega_z$ is a Cartesian product of $B$ identical constellations belonging to $2^{m'}$-$\mathcal{R}^1$ or $2^{m'}$-$\mathcal{C}^1$ and $P=I$, where $I$ is the identity matrix, so that $\Omega_x = \Omega_z$; these B-dimensional constellations contain $M=2^{m' B}$ points. 

\begin{figure}
	\centering
	\includegraphics[width=0.75 \textwidth]{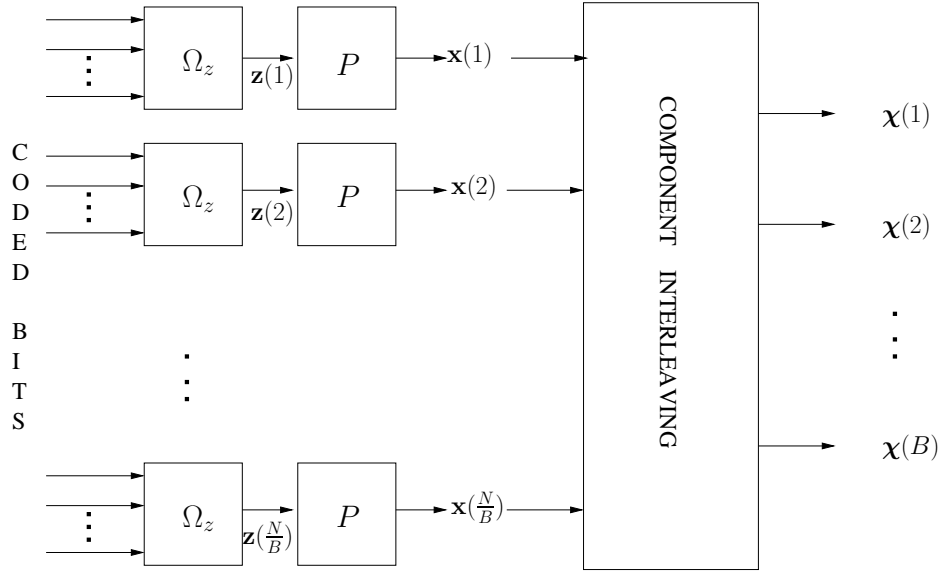}
	\caption{The coded bits are mapped to multidimensional symbols $\mb{z}(n)$ in a constellation $\Omega_z$. Each symbol $\mb{z}(n)$ is transformed to $\mb{x}(n)$, where the different components are placed in the corresponding blocks in $\bs{\chi}$.}
	\label{fig: multidimensional modulator part 1}
\end{figure}

We can reformulate Eq. (\ref{eq: system model}) in terms of $\mb{x}(n)$ as
\begin{equation}
	\label{eq: system model1b}
	\mb{y}(n) = \bs{\alpha}\cdot\mb{x}(n) + \mb{w}(n), ~~n=1, \ldots, \frac{N}{B},
\end{equation}
where $y(n)_b = \mu(b)_n$, $(\bs{\alpha})_b = \alpha_b$, $w(n)_b = \omega(b)_n$ and $\bs{\alpha}\cdot\mb{x}(n)$ denotes component-wise multiplication: $(\bs{\alpha}\cdot\mb{x}(n))_b = \alpha_b x(n)_b$.

As precoding allows to increase the spectral efficiency associated with full diversity, this technique has been extensively studied in previous works, but almost exclusively for \textit{uncoded} schemes \cite{bayer2004nac, biglieri1998fci}, \cite{boutros1998ssd}. Here, we study this system for \textit{coded} schemes and we choose the precoding matrix $P$ and the constellation $\Omega_z$ minimizing the outage probability. 

A precoding matrix $P$ that is unitary is a natural choice because it does not decrease the capacity of a Gaussian channel. 
In this paper, we restrict our study to real-valued precoding matrices, hence $P$ is orthogonal. When $B=2$, $P$ is a rotation matrix where the rotation angle  $\theta$ is the  degree of  freedom:
\begin{equation}
\label{eq: P for B=2}
	P =  \left[ \begin{array}{c c} \cos(\theta) & -\sin(\theta) \\ \sin(\theta) & \cos(\theta)  \end{array} \right].
\end{equation}
However, rotation matrices are difficult to construct for higher dimensions. In Sec. \ref{sec: Analysis of Outage Probability in the fading space}, it will be shown that for $B>2$ it is sufficient to consider orthogonal circulant precoding matrices. We denote its first row as $(p_0, \ldots, p_{B-1})$. The second row is a cyclic shift to the right of the first row, and so on. Because the columns of the $B \times B$ Fourier matrix $F$ are the eigenvectors of any circulant matrix, we can construct $P$ as follows:
\begin{equation}
\label{eq: construct fourier}
	P = F \Lambda F^H,
\end{equation}
where $(F)_{m,n} = \frac{1}{\sqrt{B}}\exp\left(\frac{-2 j \pi m n}{B} \right)$, $m,n \in \{0, \ldots, B-1 \}$, and $\Lambda$ is a diagonal matrix containing the eigenvalues of $P$. The condition for $P$ being orthogonal is $\Lambda^H \Lambda = I_B$, or the $B$ eigenvalues of $P$ must have a squared magnitude of 1. It is easy to find that 
\begin{equation}
	\lambda_n = \sum_{l=0}^{B-1}  p_{l} \exp\left(\frac{-j 2 \pi n l}{B} \right).
\end{equation}
Now, it follows that 
\begin{equation}
	p_l = \frac{1}{B} \sum_{m=0}^{B-1} \lambda_m \exp\left(\frac{j 2 \pi m l}{B} \right).
\end{equation}

As the eigenvalues must have a magnitude of 1, we have $\lambda_n = \exp(j\theta_n)$. In order to obtain a real-valued $P$, we take $\lambda_0$ real-valued (i.e., $\lambda_0 = 1$ or $\lambda_0 = -1$) and $\lambda_{B-n} = (\lambda_n)^*$ (i.e., $\theta_{B-n} = -\theta_n$) for $n = 1, \ldots, B-1$. When $B$ is even, this implies $\lambda_{B/2} = 1$ or $\lambda_{B/2} = -1$. For $B > 2$, $P$ is determined by $\lfloor(B-1)/2)\rfloor$ continuous parameters that can be optimized.

Note that for $B=3$, $P$ constructed as above, with $\lambda_0=1$ and $\lambda_1 = \exp(j \theta_1)$, corresponds to a $3$-dimensional rotation with angle $\theta_1$ around the fixed axis $\frac{1}{\sqrt{3}}(1,1,1)$,
\begin{equation}
\label{eq: P for B=3}
	P =  \frac{1}{3} \left[ \begin{array}{c c c} 
	1 + 2 k & 1-k-\sqrt{3} l & 1-k + \sqrt{3} l\\ 
	1-k +\sqrt{3} l & 1 + 2 k & 1-k-\sqrt{3} l \\
	1-k-\sqrt{3} l & 1-k + \sqrt{3} l & 1 + 2 k 
	\end{array} \right],
\end{equation}
where $k=\cos(\theta_1)$ and $l=\sin(\theta_1)$. 

Fig. \ref{fig:  rotated QAM} illustrates the effect of a rotation for $B=2$ when a $4$-$\mathcal{R}^2$ constellation is used as $\Omega_z$. The transmitted components are affected by their corresponding fading gain, which is expressed by the component-wise multiplication $\mb{t}(n) = \bs{\alpha}\cdot\mb{x}(n)$, which is shown at the right side in Fig. \ref{fig:  rotated QAM}. We say that $\mb{t}(n)$ belongs to the \textit{faded} constellation $\Omega_t$. The point $\mb{t}(n) = [t(n)_1, \ldots, t(n)_B]$ is shown for a particular fading point $\bs{\alpha}$. When $\alpha_1 \neq \alpha_2$, the constellation $\Omega_t$ can be interpreted as a distorted QPSK constellation (i.e., a constellation in which both components do not have the same magnitude). 

\begin{figure}
	\centering
	\subfigure[]{\includegraphics[width=0.45 \textwidth]{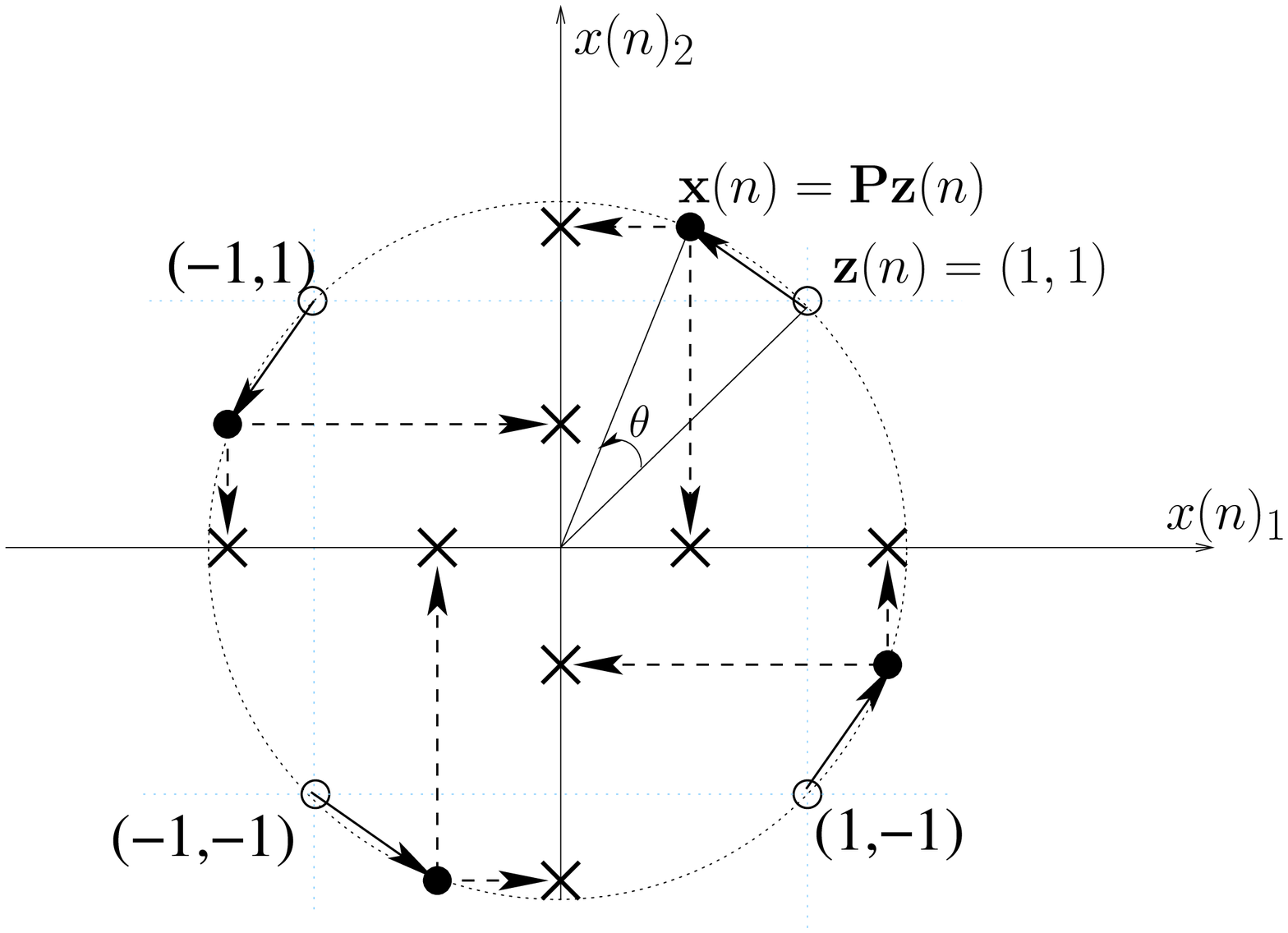}\label{fig: balanced}}
	\quad 
	 \subfigure[]{\includegraphics[width = 0.45\textwidth]{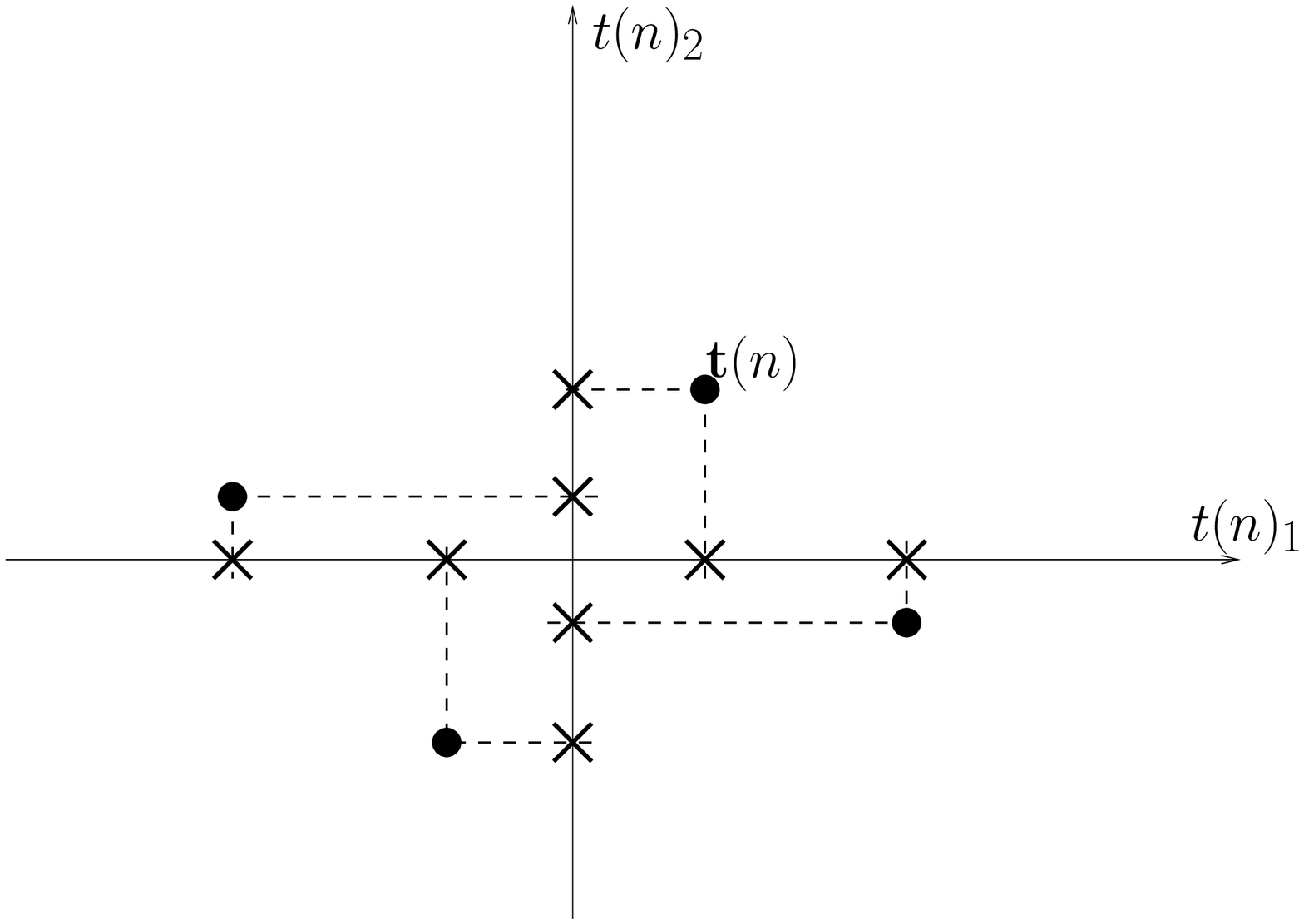}\label{fig: unbalanced}}
	\caption{Displaying the rotation at the transmitter \subref{fig: balanced} and at the receiver (without noise) \subref{fig: unbalanced} for $B=2$. 
The empty (filled) circles represent $\Omega_z$ ($\Omega_x$). The components of $\mb{t}(n)$ are obtained by scaling the components $\mb{x}(n)$ by their respective fading gain (here $\alpha_2 < \alpha_1$), as expressed by the component-wise multiplication $\bs{\alpha} \cdot \mb{x}(n)$. The crosses on the coordinate axes are the transmitted \subref{fig: balanced} and received \subref{fig: unbalanced} vector components, respectively.}
	\label{fig: rotated QAM}
\end{figure}

Considering $\mb{t}(n)$, an equivalent channel model can be formulated:
\begin{equation}
	\label{eq: system model2}
	\mb{y}(n) = \mb{t}(n) + \mb{w}(n), ~~n=1, \ldots, \frac{N}{B},
\end{equation}

which yields more insight and will be useful in the proofs of propositions of this paper. This system model represents a Gaussian vector channel with input $\mb{t}(n)$. This means that for a particular fading point, the block fading channel can be interpreted as a virtual\footnote{We use the term \textit{virtual} because the fading gains of the actual channel are incorporated in the constellation $\Omega_t$.} Gaussian channel, with a discrete input alphabet $\Omega_t$.

\section{Analysis of Outage Probability in the fading space}
\label{sec: Analysis of Outage Probability in the fading space}

For the remainder of the paper, we will drop the index $n$ in the vectors $\mb{z}(n)$, $\mb{x}(n)$, $\mb{t}(n)$, $\mb{y}(n)$ and $\mb{w}(n)$, as the time index is not important when considering mutual information. We write random variables using upper case letters corresponding to the lower case letters used for their realizations. The mutual information $I\left(\bs{\alpha}, \gamma,  P \right)$ at a certain fading point $\bs{\alpha}$ between the transmitted B-dimensional symbol $\mb{x}$ (uniformly distributed over $\Omega_x$) and the corresponding received vector $\mb{y}$ is given by
\begin{equation}
	I\left(\bs{\alpha}, \gamma, P \right) =  \frac{1}{B} I(\mb{X};\mb{Y}|\bs{\alpha}, \gamma), 
\label{eq.: mutual info}
\end{equation}
where the scaling factor $\frac{1}{B}$ is added because the $B$ blocks in the channel timeshare a time-interval \cite[Section 9.4]{cover2006eit}, \cite[Section 5.4.4]{tse2005fwc}. The mutual information $I(\mb{X};\mb{Y}|\bs{\alpha}, \gamma)$ is \cite{fab2007mcm}, \cite{fabregas2004ccb}
\begin{equation}
	I(\mb{X};\mb{Y}|\bs{\alpha}, \gamma) =  m - 2^{-m} \sum_{\mb{x} \in \Omega_x}  \E_{\mb{y}|\mb{x}} \left[  \Log_2 \left( \sum_{\mb{x}^{'} \in \Omega_x} \exp\left[ \frac{d^2(\mb{y},\bs{\alpha}\cdot\mb{x}) - d^2(\mb{y}, \bs{\alpha}\cdot\mb{x}')}{2 \sigma^2} \right] \right) \right],
\label{eq.: mutual info 2}
\end{equation}
where $d^2(\mb{v},\mb{u})= \sum_{b=1}^B \left|v_b-u_b \right|^2$. However, more insight can be gained when considering $\mb{t}$. From Eq. (\ref{eq: system model2}), it is clear that for a certain fading point, the mutual information of this virtual channel, with input $\Omega_t$, is the same as the mutual information of the actual channel, with input $\Omega_x$,
\begin{equation}
	\label{eq: mutual information tn}
	I\left(\bs{\alpha}, \gamma,  P \right) =  \frac{1}{B} I(\mb{T};\mb{Y}|\bs{\alpha}, \gamma).
\end{equation}
Hence, the fading point $\bs{\alpha}$ maximizing (minimizing) the mutual information $I(\mb{X};\mb{Y}|\bs{\alpha}, \gamma)$ corresponds to the fading point that distorts constellation $\Omega_t$ in the best (worst) way at the input of a Gaussian vector channel. This interpretation allows to exploit the many results from literature on the mutual information of Gaussian channels. Therefore, Eq. (\ref{eq: mutual information tn}) will be useful in the following. 

The outage probability is the probability that the instantaneous mutual information is less than the transmitted information bitrate $R$ \cite[section 5.4]{tse2005fwc},
\begin{equation}
\label{eq: outage probability}
P_{\textrm{out}}(\gamma, P, R) = \Prob\big(I\left(\bs{\alpha}, \gamma, P\right) < R \big).
\end{equation}
Our main goal is to find the precoding matrix $P$ that minimizes the outage probability $P_{\textrm{out}}(\gamma, P, R)$,
\begin{equation*}
P_{\textrm{opt}} = \argmin{P}~ P_{\textrm{out}}(\gamma, P, R) \label{eq: research question}.
\end{equation*}

A  closed   form   expression  for $I\left(\bs{\alpha},   \gamma,   P\right)$   does   not   exist 
and Eq. (\ref{eq.:  mutual   info 2})  is   difficult  to  analyze   because  of the presence of the mathematical expectation. Therefore, in the following two sections, we will develop simple bounds on the outage probability and will approximate $P_{\textrm{opt}}$ by the optimal precoding matrix that minimizes  the bounds on the outage probability.  

But first we will introduce a framework in this section that allows to gain insight on the meaning of the outage probability. The considered framework is the fading space \cite{bou2005aoc}, which is the B-dimensional Euclidean space, ${\mathbb{R}^+}^B$, of the real-valued positive fading gains. The outage probability equals the probability that the fading point $\bs{\alpha}$ belongs to a region $V_{\textrm{o}}$, which is such that $I\left(\bs{\alpha}, \gamma, \theta\right)  <  R$ for all $\bs{\alpha}$ in $V_{\textrm{o}}$ (Fig. \ref{fig: example outage boundary}):
\begin{equation}
	P_{\textrm{out}}(\gamma, P, R)=\int_{\bs{\alpha} \in V_{\textrm{o}}} p(\bs{\alpha}) \mathrm{d}\bs{\alpha},
\label{eq.: outage prob in space}
\end{equation}
where $p(\bs{\alpha})$ is the joint pdf of the fading gains $\alpha_1, \ldots, \alpha_B$.  We  say  that  the region  $V_{\textrm{o}}$  is  limited  by  an \textit{outage boundary} $B_{\textrm{o}}(\gamma, P, R)$ (see Fig. \ref{fig: example outage boundary} for $B=2$), defined by 
\[ I\left(\bs{\alpha}, \gamma, P\right) = R,~~ \forall~ \bs{\alpha} \in B_{\textrm{o}}(\gamma, P, R). \]
\begin{figure}
	\centering
	\subfigure[Examples of an outage boundary when $\Omega_z=$ $4$-$\mathcal{R}^2$ and the rotation angle is $\theta=0, 10, 27$ degrees. The region $V_{\textrm{o}}$ is coloured red for $\theta=27$ degrees.]{\includegraphics[width=0.40 \textwidth, angle=-90]{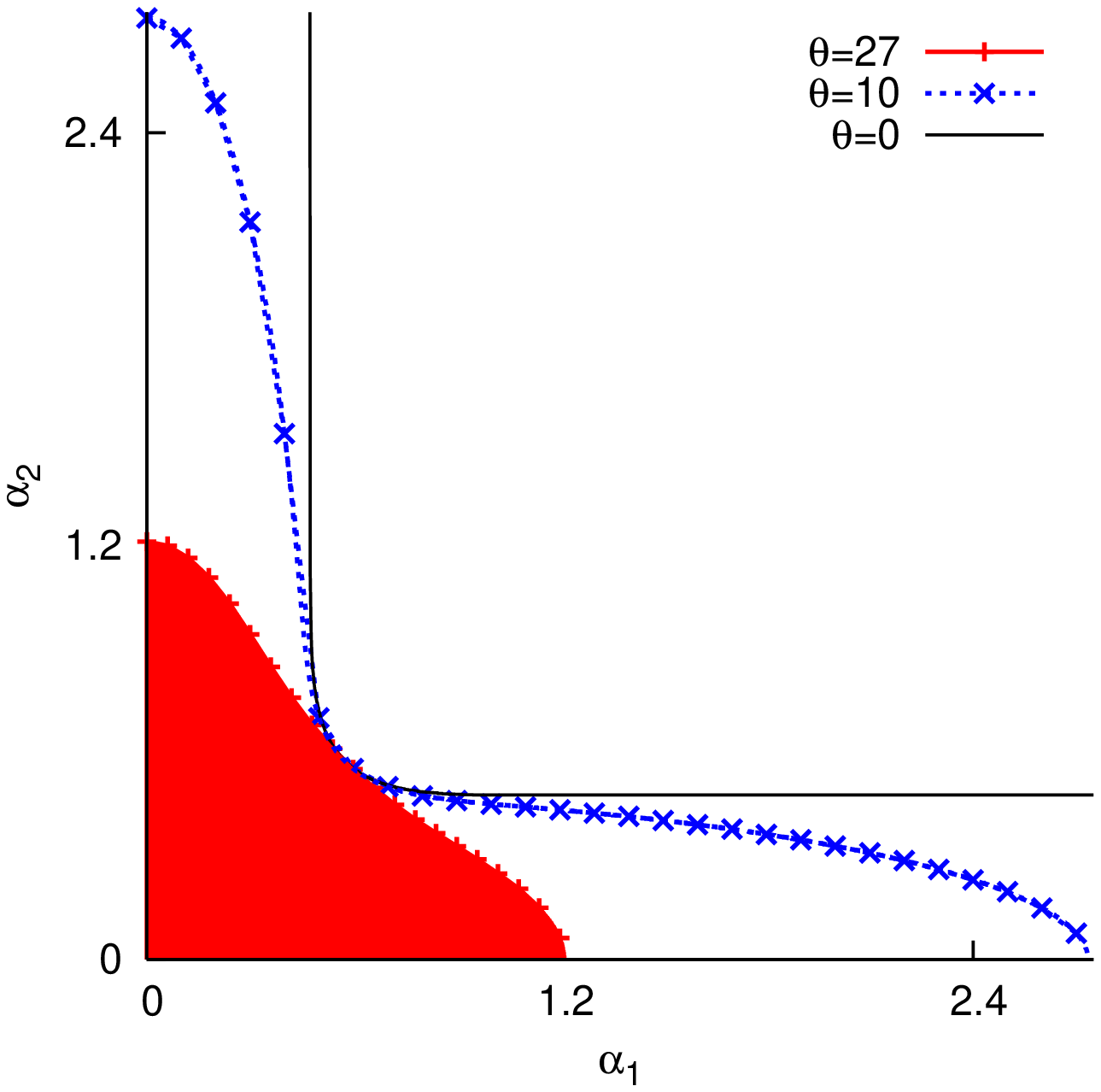}\label{fig: example outage boundary}}
	\quad 
	 \subfigure[The points $\{\alpha_{b,\textrm{o}}$, $b=1, \ldots, B\}$ and $\alpha_{\textrm{e}}$ are shown for $B=2$.]{\includegraphics[width = 0.40\textwidth, angle=-90]{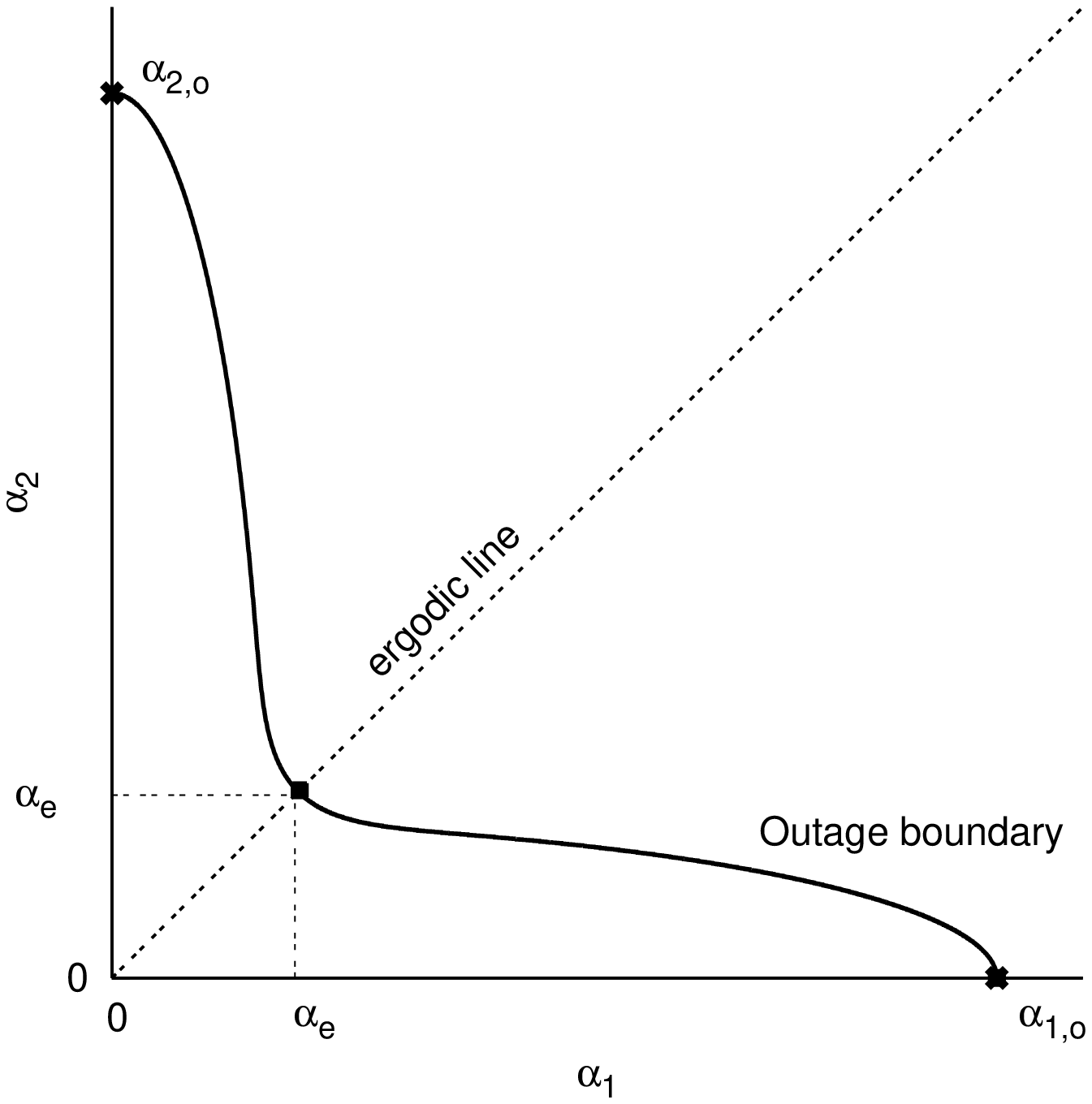}\label{fig: out bound}}
	\caption{The outage boundary limits the outage region $V_{\textrm{o}}$ in the fading space which corresponds to an information theoretic outage event. In the Figure, the information rate is $R=0.9$ bpcu and the average SNR is fixed to $\gamma=8$ dB.}
	\label{fig: outage boundary}
\end{figure}
In general, we can consider a certain boundary in the fading space limiting a certain region $V$. This boundary is denoted by $B(V)$. The outage boundary $B_{\textrm{o}}(\gamma, P, R)$ corresponds to the region $V_{\textrm{o}}$, as defined in (\ref{eq.: outage prob in space}), so that the outage boundary is $B_{\textrm{o}}(\gamma, P, R) = B(V_{\textrm{o}})$.

\begin{definition}
\label{def: alpha_bo}
We define $\alpha_{b,\textrm{o}}$ by the magnitude of the intersection between the outage boundary and the axis $\alpha_b$. More precisely,
$I\left(\bs{\alpha}\arrowvert_{\alpha_i=0, i \neq b, \alpha_b=\alpha_{b,\textrm{o}}}, \gamma, P \right) = R$. 
By convention, $\alpha_{b,\textrm{o}}=+\infty$ if the axis $\alpha_b$ is an asymptote for the outage boundary. It is possible that $\alpha_{b,\textrm{o}}$ does not exist.
\end{definition}

\begin{definition}
\label{def: alpha_e}
We define $\alpha_{\textrm{e}}$ as the value of the components of the intersection between the outage boundary and the line $\alpha_1=\ldots=\alpha_B$ (also known as the {\em ergodic line}). More precisely, $I\left(\bs{\alpha}\arrowvert_{\alpha_i=\alpha_{\textrm{e}}, i \in [1,\ldots, B]}, \gamma, P \right) = R$.
\end{definition}
The defined points are illustrated in Fig. \ref{fig: out bound} for $B=2$. In the remainder of the paper, we denote the points $\bs{\alpha}\arrowvert_{\alpha_i=0, i \neq b, \alpha_b=\alpha_{b,\textrm{o}}}$ by $\bs{\alpha}_{b,\textrm{o}}$ and $\bs{\alpha}\arrowvert_{\alpha_i=\alpha_{\textrm{e}}, i \in [1,\ldots, B]}$ by $\bs{\alpha}_{\textrm{e}}$.

\begin{proposition}
\label{Prop: symmetry}
On a BF channel with $B=2$, with a discrete input alphabet and with linear precoding, the magnitudes $\{\alpha_{b,\textrm{o}},~ b=1, 2\}$ are equal if the constellation is invariant under a rotation of 90 degrees and the precoding matrix is orthogonal.\\
On a BF channel with $B>2$, with a discrete input alphabet and with linear precoding, the magnitudes $\{\alpha_{b,\textrm{o}},~ b=1, \ldots, B\}$ are equal if the constellation is invariant under a cyclic shift of the components of the points of the constellation and the precoding matrix is an orthogonal circulant matrix.
\end{proposition}
\begin{proof}
See Appendix \ref{appendix: symmetry}.
\end{proof}
Notice that the condition of Prop. 1 is a sufficient condition and not a necessary condition. In the remainder of this paper, it is assumed that the constellation used at the transmitter fulfils Prop. \ref{Prop: symmetry}. The magnitudes $\{ \alpha_{b,\textrm{o}}, b=1,\ldots, B \}$ will then simply be denoted by $\alpha_{\textrm{o}}$. This also means that the projection of $\Omega_x$ on either coordinate axes yields the same set of points, which we denote by $\mathcal{S}_{\textrm{p}}$, where $p$ stands for projection. 

Note that a multidimensional constellation that fulfils Prop. \ref{Prop: symmetry} (i.e., its projection on each coordinate axis yields the same set of points, see Appendix \ref{appendix: symmetry}) has an interesting property. For these constellations, the function $I\left(X_b; Y_b | \alpha_b = \alpha ,\gamma\right)$, which is the mutual information of a point-to-point channel with fading coefficient $\alpha_b =\alpha$, average SNR $\gamma$ and with discrete input $X_b$, does not depend on $b$. As a consequence, we will represent this mutual information by $I_{\mathcal{S}_{\textrm{p}}}\left(\alpha^2 \gamma, P\right)$.

\begin{definition}
\label{def: upper bound}
Boundary $B_1=B(V_1)$ is said to outer bound outage boundary $B_2=B(V_2)$
if $V_2 \subseteq V_1$. The opposite holds for inner bounding.
\end{definition}

\noindent The outage boundary in the fading space has a simple but interesting property:
if an outage boundary \textit{outer bounds (inner bounds)} another outage boundary, than its corresponding outage probability is larger (smaller) (see Eq. (\ref{eq.: outage prob in space})). For example, consider the input alphabet $\mb{Z} \sim \mathcal{N}(\mb{0}, I)$, which we denote as an i.i.d. Gaussian input alphabet\footnote{We mainly refer to real-valued i.i.d. Gaussian input alphabets in this paper. In Sec. \ref{sec: optimization complex}, we extend this to complex  i.i.d. Gaussian input alphabets.}, where $\mb{0}$ is a column vector of zeros and $I$ is the identity matrix. We denote the outage boundary (outage region) corresponding to i.i.d. Gaussian inputs and discrete input alphabets, by $B_{\textrm{o}}\textrm{(Gauss)}$ ($V_{\textrm{o}}\textrm{(Gauss)}$) and $B_{\textrm{o}}\textrm{(discrete)}$ ($V_{\textrm{o}}\textrm{(discrete)}$), respectively. The boundary $B_{\textrm{o}}\textrm{(Gauss)}$ inner bounds $B_{\textrm{o}}\textrm{(discrete)}$ \cite{cover2006eit}. Therefore, the outage probability corresponding to i.i.d. Gaussian inputs is a lower bound on the outage probability corresponding to a discrete input alphabet. Consequently, by minimizing the outage probability corresponding to a discrete input alphabet, $B_{\textrm{o}}\textrm{(discrete)}$ can at most approach $B_{\textrm{o}}\textrm{(Gauss)}$ (see Fig. \ref{fig: out bound approach Gauss}). 

\begin{figure}[!htb]
   \centering
   \includegraphics[width=0.6 \textwidth, angle=-90]{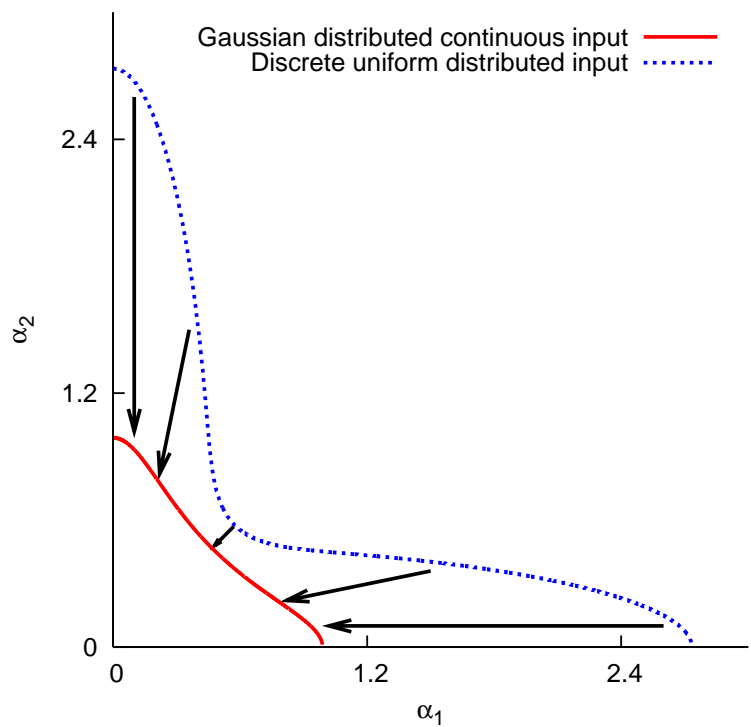} 
 	\caption{The outage boundary $B_{\textrm{o}}\textrm{(discrete)}$ is inner bounded by $B_{\textrm{o}}\textrm{(Gauss)}$. Optimizing the precoding matrix can at most make $B_{\textrm{o}}\textrm{(discrete)}$ approach $B_{\textrm{o}}\textrm{(Gauss)}$ as illustrated by the arrows. }  
	\label{fig: out bound approach Gauss}
\end{figure}

In the following two sections, we will determine boundaries with simple shapes outer (inner) bounding $B_{\textrm{o}}\textrm{(discrete)}$, which are then much easier to optimize. For the outer boundary, this can be done by determining a surface in the fading space, $U(\bs{\alpha})=0$, satisfying
\begin{equation}
\label{eq: upper bounding function}
I\left(\bs{\alpha}, \gamma, P\right) \geq R, ~\textrm{for all } \bs{\alpha} \textrm{ satisfying } U(\bs{\alpha})=0.
\end{equation}
For the inner boundary, the greater than or equal sign is replaced by a smaller than or equal sign. For example, for $B=2$, we will prove that a circular arc touching $B_{\textrm{o}}\textrm{(discrete)}$ at  $\alpha_1=\alpha_{\textrm{o}}$  (on the horizontal axis) and $\alpha_2=\alpha_{\textrm{o}}$  (on the vertical axis) satisfies Eq. (\ref{eq: upper bounding function}). 

\section{Bounds on the outage probability without linear precoding}
\label{sec: Bounds on the outage probability without linear precoding}

As a first step, we will establish upper and lower bounds on the outage probability of a communication system without linear precoding. This will set the stage for the bounds with linear precoding in Sec. \ref{sec: linear precoding}. 

We will prove that the outage region is outer bounded by a $B$-hypersphere touching the outage boundary on the axes of the fading space, hence with radius $\alpha_{\textrm{o}}$. A $B$-hypersphere $U(\bs{\alpha})=0$ is a generalization of a sphere to $B$ dimensions,
\[ U(\bs{\alpha}) = \sum_{b=1}^B \alpha_b^2 - \alpha_{\textrm{o}}^2. \] Note that this is only possible for constellations fulfilling the conditions of Prop. \ref{Prop: symmetry}. For other constellations, the outer bounding $B$-hypersphere will have a radius of $\max_{b \in [1, \ldots, B]}\alpha_{b,\textrm{o}}$ and will therefore be less tight. 

Next, we will prove that the outage region is inner bounded by a $B$-hypersphere touching the outage boundary at $\bs{\alpha}_{\textrm{e}}$.

\begin{lemma}
\label{Lemma: Cartesian}
On a BF channel with a discrete multidimensional input alphabet $\mb{X}$, the mutual information $I\left(\bs{\alpha},\gamma,P \right)$ is upper bounded as follows 
\begin{equation}
	I\left(\bs{\alpha},  \gamma,  P \right)  \leq \frac{1}{B}\sum_{b=1}^B I_{\mathcal{S}_{\textrm{p}}}\left(\alpha_b^2 \gamma, P\right).
\label{eq.: mutual info theta}
\end{equation}
\end{lemma}
\begin{proof}
See Appendix \ref{app: lemma}.
\end{proof}

A special case is the case without precoding ($P=I$) so that $\Omega_x = \Omega_z$. If in that case $\Omega_z$ is the Cartesian product of $\Omega_{z_1}$, \ldots, $\Omega_{z_B}$ \footnote{We denote the projection of $\Omega_z$ on the $b$-th coordinate axis as $\Omega_{z_b}$. For constellations fulfilling Prop. \ref{Prop: symmetry}, $\Omega_{z_b} = \mathcal{S}_{\textrm{p}}$.}, $\Omega_z = (\mathcal{S}_{\textrm{p}})^B$, then all variables $X_b=Z_b$ are independent for all $b$, so that equality holds in Eq. (\ref{eq.: mutual info theta}). E.g., $\Omega_z =$ $4$-$\mathcal{R}^2$, as shown in Fig. \ref{fig: rotated QAM}, is the Cartesian product of two BPSK constellations, denoted as BPSK$\times$BPSK. 

For an i.i.d. Gaussian input alphabet, the instantaneous mutual information with linear precoding is the same as the instantaneous mutual information without linear precoding, because $P \mb{z} = \mb{x} \sim \mathcal{N}(\mb{0}, I)$ when $P$ is orthogonal, so that precoding does not change the distribution of the input vector. Therefore, the instantaneous mutual information for an i.i.d. Gaussian input alphabet is
\begin{equation}
	I\left(\bs{\alpha}, \gamma,  P\right) = I\left(\bs{\alpha},  \gamma, P=I\right) = \frac{1}{B} \sum_{b=1}^B 0.5 ~ \Log_2(1+ 2 \gamma \alpha_b^2).
\label{eq.: mutual info Gauss}
\end{equation}

\begin{proposition}
\label{bound: gauss}
The outage region $V_{\textrm{o}}\textrm{(Gauss)}$ of a BF channel with an i.i.d. Gaussian input alphabet is outer bounded by the $B$-hypersphere $\alpha_1^2+\alpha_2^2 + \ldots + \alpha_B^2 = \frac{4^{BR}-1}{2 \gamma}$ and  inner  bounded  by  the $B$-hypersphere  $\alpha_1^2+\alpha_2^2 + \ldots + \alpha_B^2  = \frac{4^R-1}{2 \gamma}$.
\end{proposition}
\begin{proof}
In Appendix \ref{app: Proof of Prop. bound: gauss}, it is proved that for $\bs{\alpha}$ on a hypersurface:
\begin{itemize}
	\item $I\left(\bs{\alpha}, \gamma\right)\arrowvert_{\alpha_1=\ldots=\alpha_B} \geq I\left(\bs{\alpha}, \gamma\right)$ 
	\item $I\left(\bs{\alpha}, \gamma\right)\arrowvert_{\alpha_b=\alpha_{\textrm{o}}} \leq I\left(\bs{\alpha}, \gamma\right)$,
\end{itemize}
because the mutual information is a sum of concave functions of the instantaneous SNRs $\gamma \alpha_b^2$ (Eq. (\ref{eq.: mutual info Gauss})). Calculating $\alpha_{\textrm{o}}$ and $\alpha_{\textrm{e}}$, according to Definitions \ref{def: alpha_bo} \& \ref{def: alpha_e}, yields the radii of both $B$-hyperspheres.
\end{proof}
The inner boundary touches the outage boundary in the point $\bs{\alpha}_{\textrm{e}}$, so that it does not depend on $B$ (the entire codeword is affected by the same fading gain $\alpha_{\textrm{e}}$). The outer boundary touches the outage boundary in the points $\bs{\alpha}_{b,\textrm{o}}$, $\forall ~b$; its dependence on $B$ follows from the fact that $B-1$ is equal to the number of erased channels. 

\begin{proposition}
\label{bound: bpsk}
On a BF channel with a discrete input alphabet $\Omega_z$ that is a Cartesian product of one-dimensional constellations, the outage region $V_{\textrm{o}}\textrm{(discrete)}$ is outer bounded by the $B$-hypersphere $\alpha_1^2+\alpha_2^2 + \ldots + \alpha_B^2 = \alpha_{\textrm{o}}^2$ touching it at the axes of the fading space at $\bs{\alpha}_{b,\textrm{o}}$, $\forall ~b$. Also, $V_{\textrm{o}}\textrm{(discrete)}$ is inner bounded by the $B$-hypersphere $\alpha_1^2+\alpha_2^2 + \ldots + \alpha_B^2  = B\alpha_{\textrm{e}}^2$ touching it at $\bs{\alpha}_{\textrm{e}}$.
\end{proposition}
\begin{proof}
Lemma \ref{Lemma: Cartesian} proved that the mutual information is upper bounded by $\frac{1}{B}\sum_{b=1}^B I_{\mathcal{S}_{\textrm{p}}}\left(\alpha_b^2 \gamma, P\right)$, where this upper bound coincides with the exact expression in the case that $\Omega_z$ is a Cartesian product. Using the relation between the mutual information $I(\textrm{SNR}) = I_{\mathcal{S}_{\textrm{p}}}\left(\alpha_b^2 \gamma, P\right)$ and the minimum mean-square error (MMSE) in estimating the input symbol, $X_b \in \mathcal{S}_{\textrm{p}}$, given the output symbol $Y_b$ \cite{guo2005it},
\begin{equation}
\label{eq: mmse I}
\frac{\textrm{d}}{\textrm{d SNR}} I(\textrm{SNR}) = \frac{1}{2} \textrm{MMSE}(\textrm{SNR}), 
\end{equation}
it is easily proved that $I_{\mathcal{S}_{\textrm{p}}}\left(\alpha_b^2 \gamma, P\right)$ is a concave function (the second derivative is negative) of the instantaneous SNR $\gamma \alpha_b^2$, because the MMSE is a decreasing function of the SNR. Therefore, the proof is the same as for Prop. \ref{bound: gauss}. 
\end{proof}

Notice that the techniques of the proofs of Props. \ref{bound: gauss} and \ref{bound: bpsk} cannot be used for discrete input alphabets $\Omega_z$ that are not a Cartesian product of one-dimensional constellations, because in that case the upper bound $\frac{1}{B}\sum_{b=1}^B I_{\mathcal{S}_{\textrm{p}}}\left(\alpha_b^2 \gamma, P\right)$ does not coincide with the exact expression of the mutual information. However, this case is merely a particular case ($P=I$) of a \textit{precoded} discrete input alphabet, which is covered in the next section. 

\section{Bounds on the outage probability with linear precoding}
\label{sec: linear precoding}

Propositions \ref{bound:  gauss} and \ref{bound: bpsk} mainly state that the outage  region for a  channel with  an i.i.d. Gaussian input alphabet or  a discrete input alphabet without precoding is outer (inner) bounded by a $B$-hypersphere touching it at $\bs{\alpha}_{b,\textrm{o}}, ~ \forall ~b$ ($\bs{\alpha}_{\textrm{e}}$). We conjecture that this property still  holds for a communication system with a discrete alphabet with linear precoding at the input of the channel. First, we will give new detailed proofs for low and high instantaneous SNR of this property. Then, a more intuitive explanation will be given to provide more insight.   

\begin{proposition}
\label{bound: bpsk with precoding low SNR}
On a BF channel at low instantaneous SNR, with a discrete input alphabet and with linear precoding, the outage region $V_{\textrm{o}}\textrm{(discrete)}$ is outer (inner) bounded by the $B$-hypersphere $\alpha_1^2+\alpha_2^2 + \ldots + \alpha_B^2 = \alpha_{\textrm{o}}^2$ ($\alpha_1^2+\alpha_2^2 + \ldots + \alpha_B^2 = B \alpha_{e}^2$). 
\end{proposition}
\begin{proof}
See Appendix \ref{app: proof prop. precoding low SNR}.
\end{proof}

\begin{proposition}
\label{bound: bpsk with precoding high SNR}
On a BF channel at high instantaneous SNR, with a discrete input alphabet and with linear precoding, the outage region $V_{\textrm{o}}\textrm{(discrete)}$ is outer (inner) bounded by the $B$-hypersphere $\alpha_1^2+\alpha_2^2 + \ldots + \alpha_B^2 = \alpha_{\textrm{o}}^2$ ($\alpha_1^2+\alpha_2^2 + \ldots + \alpha_B^2 = B \alpha_{e}^2$). 
\end{proposition}
\begin{proof}
See Appendix \ref{app: proof prop. precoding high SNR}.
\end{proof}

The outer and inner bound touch the outage boundary at the axes of the fading space and on the ergodic line, respectively. The outer (inner) region corresponds to an upper (lower) bound on the outage probability. Minimizing this upper (lower) bound is simply achieved by minimizing $\alpha_{\textrm{o}}$ ($\alpha_{e}$). 

%
%

To get more insight, we give an illustration of Props. \ref{bound: bpsk with precoding low SNR} and \ref{bound: bpsk with precoding high SNR}. We consider the $4$-$\mathcal{R}^2$ constellation $\Omega_x$ from Fig. \ref{fig: balanced}, with $\theta = \pi/6$. Further, we take $\alpha_1 = \alpha_{\textrm{o}} \cos(\lambda)$ and $\alpha_2 = \alpha_{\textrm{o}} \sin(\lambda)$, i.e., $|\bs{\alpha}| = \alpha_{\textrm{o}}$, so that $\bs{\alpha}$ is on the outer boundary of the outage boundary. Fig. \ref{fig:  unbalanced QPSK} shows the corresponding mutual information $I(\bs{\alpha},\gamma,\theta)$ as a function of $\gamma$ for various values of $\lambda$. We observe that the mutual information increases when $\lambda$ increases from $0$ to $\pi/4$. The projections of the constellation $\Omega_t$ on the horizontal and vertical coordinate axes have variances $\alpha_1^2 = \alpha_{\textrm{o}}^2 \cos^2(\lambda)$ and $\alpha_2^2 = \alpha_{\textrm{o}}^2 \sin^2(\lambda)$ respectively, whereas the total variance of $\Omega_t$ equals $\alpha_{\textrm{o}}^2$, irrespective of $\lambda$. When $\lambda = \pi/4$, $\Omega_t$ is equivalent to a common QPSK constellation: the variances of both projections equal $\frac{\alpha_{\textrm{o}}^2}{2}$, and the corresponding mutual information is maximum. When $\lambda$ decreases from $\pi/4$ to $0$, the difference between the two variances increases, $\Omega_t$ becomes an increasingly more distorted QPSK constellation, and the mutual information decreases. At $\lambda = 0$ we get $\alpha_1 = \alpha_{\textrm{o}}$ and $\alpha_2 = 0$; $\Omega_t$ (which is now a scaling of $\mathcal{S}_{\textrm{p}}$) reduces to a 4-PAM constellation and the corresponding mutual information is minimum. This illustrates the general principle that performance is optimized when the transmit power is equally split over the available identical channels and the performance is worst when the transmit power is completely used for only one channel, which is exactly what is claimed in Props. \ref{bound: bpsk with precoding low SNR} and \ref{bound: bpsk with precoding high SNR} for high and low instantaneous SNR. More specifically, for $\bs{\alpha}$ on the hypersurface of a hypersphere, it is proved in Props. \ref{bound: bpsk with precoding low SNR} and \ref{bound: bpsk with precoding high SNR} that the mutual information is the smallest at $\bs{\alpha}_{b,\textrm{o}}$, and the largest at $\bs{\alpha}_{e}$.
\begin{figure}
	\centering
	\includegraphics[angle=-90, width=0.8 \textwidth]{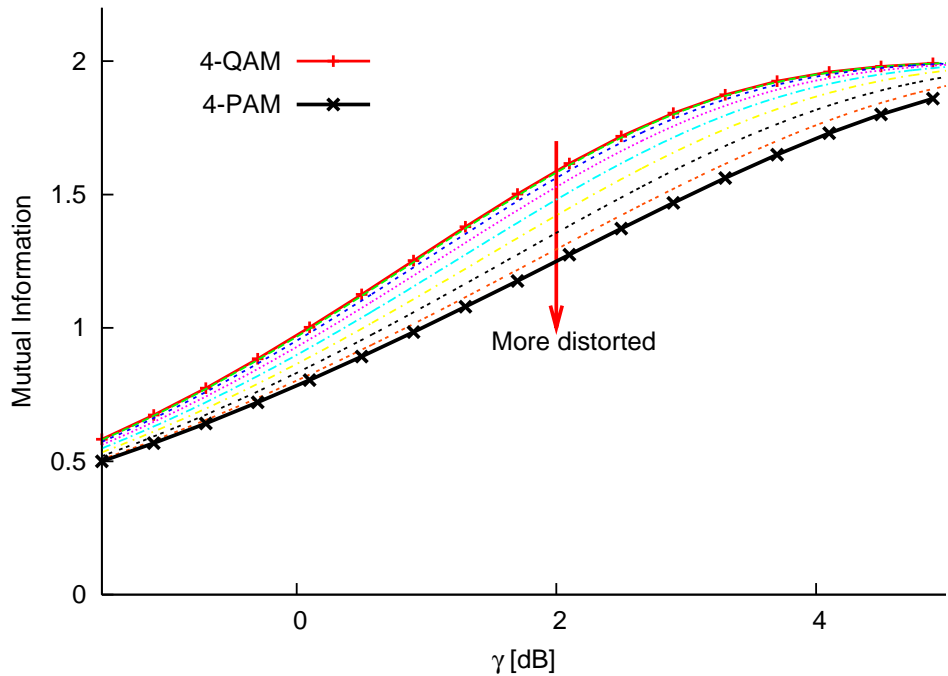}
	\caption{The mutual information is monotonically non-increasing as the constellation becomes more distorted. The upper and lower curve correspond to $\lambda=\pi/4$ and $\lambda=0$ respectively.}
	\label{fig: unbalanced QPSK}
\end{figure}


Prop. \ref{prop: full diversity} states the condition under which the precoded constellations achieve full diversity for a rate $R$. This is proved in \cite{fab2007mcm}, but based on Prop. \ref{bound: bpsk with precoding high SNR}, a new simple proof can be given.

\begin{proposition}
\label{prop: full diversity}
For any coding rate $0 < R_c < 1$, the outage probability of a Rayleigh distributed block fading channel with $B$ blocks, with a discrete input alphabet and linear precoding exhibits full diversity, i.e., $P_{\textrm{out}}(\gamma, P, R)  \propto 1/\gamma^B$ at large $\gamma$, if $2^{B R}$ does not exceed the number of points contained in the projection of the precoded input constellation $\Omega_x$ on a coordinate axis. 
\end{proposition}
\begin{proof}
See Appendix \ref{app: prop. full diversity}.
\end{proof}

It should be noted that the proof of Prop. \ref{prop: full diversity} assumes the existence of $I_{\mathcal{S}_{\textrm{p}}}^{-1}(B R, P)$; hence, the number of points in $\mathcal{S}_{\textrm{p}}$ must not be less than $2^{B R}$. We consider the $4$-$\mathcal{R}^2$ constellation from Fig. \ref{fig: balanced} to illustrate the constraint on the existence of $I_{\mathcal{S}_{\textrm{p}}}^{-1}(B R, P)$. For $4$-$\mathcal{R}^2$ from Fig. \ref{fig: balanced}, $\mathcal{S}_{\textrm{p}}$ contains only $2$ points when the rotation angle $\theta$ is a multiple of $\pi/2$, $3$ points when the rotation angle $\theta$ is a multiple of $\pi/4$ but not of $\pi/2$, and $4$ points otherwise, yielding $R=0.5$, $R=3/4$ and $R=1$, respectively, as maximal rates corresponding to full diversity.

\section{Optimization of the outage probability of precoded constellations}
\label{sec: optimization real}

In the previous section, we established that, for high and low instantaneous SNR, the outage region of block fading channels with  precoded constellations is outer bounded  by a $B$-hypersphere with center in the origin and touching the outage boundary on the axes of the fading space. This hypersphere corresponds to an upper bound on the outage probability of the channel (see the paragraph after Def. \ref{def: upper bound}). Instead of minimizing the actual outage probability, it is easier to minimize the upper bound on the outage probability. This optimization will allow the actual outage probability to closely approach a lower bound on the outage probability, i.e., the outage probability corresponding to an i.i.d. Gaussian input alphabet, as we will see in the numerical results.

The $B$-hypersphere is completely determined by one variable, its radius $\alpha_{\textrm{o}}$. We denote the $B$-hypersphere outer bounding the outage region by $V_{\textrm{up}}(\alpha_{\textrm{o}})$ and the corresponding upper bound on the outage probability by $P_{\textrm{up}}(\alpha_{\textrm{o}})$:
\begin{equation}
	P_{\textrm{up}}(\alpha_{\textrm{o}})=\int_{\bs{\alpha} \in V_{\textrm{up}}(\alpha_{\textrm{o}})} p(\bs{\alpha)} \mathrm{d}\bs{\alpha}.
\label{eq.: upperbound outage prob in space}
\end{equation}
From Eq. (\ref{eq.: upperbound outage prob in space}), it is clear that the region $V_{\textrm{up}}(\alpha_{\textrm{o}})$ has to be made as small as possible to minimize $P_{\textrm{up}}(\alpha_{\textrm{o}})$. Therefore, the optimization target is to minimize the radius $\alpha_{\textrm{o}}$. 

Because $I_{\mathcal{S}_{\textrm{p}}}(\alpha_{\textrm{o}}^2 \gamma, P) = B R$, the minimization of $\alpha_{\textrm{o}}$ (and, therefore, the minimization of the upper bound on the outage probability) is achieved by selecting the constellation $\mathcal{S}_{\textrm{p}}$ requiring the least energy to achieve a rate $BR$ on a Gaussian channel. This involves a proper selection of both the constellation $\Omega_z$ and the precoding matrix $P$. Note that this optimization is much simpler than the direct minimization of the outage probability as given by Eqs. (\ref{eq: outage probability}) and (\ref{eq.: mutual info}), because the outage probability is hard to evaluate, especially when the number of fading gains and constellation points is large. Furthermore, no insight is gained by the latter approach, so that it would not be clear which constellation $\Omega_z$ should be taken.

\begin{figure}
	\centering
	\subfigure[The transmitted symbols are the components $x_b$ of the constellation $\Omega_x$.]{\includegraphics[width=0.40 \textwidth]{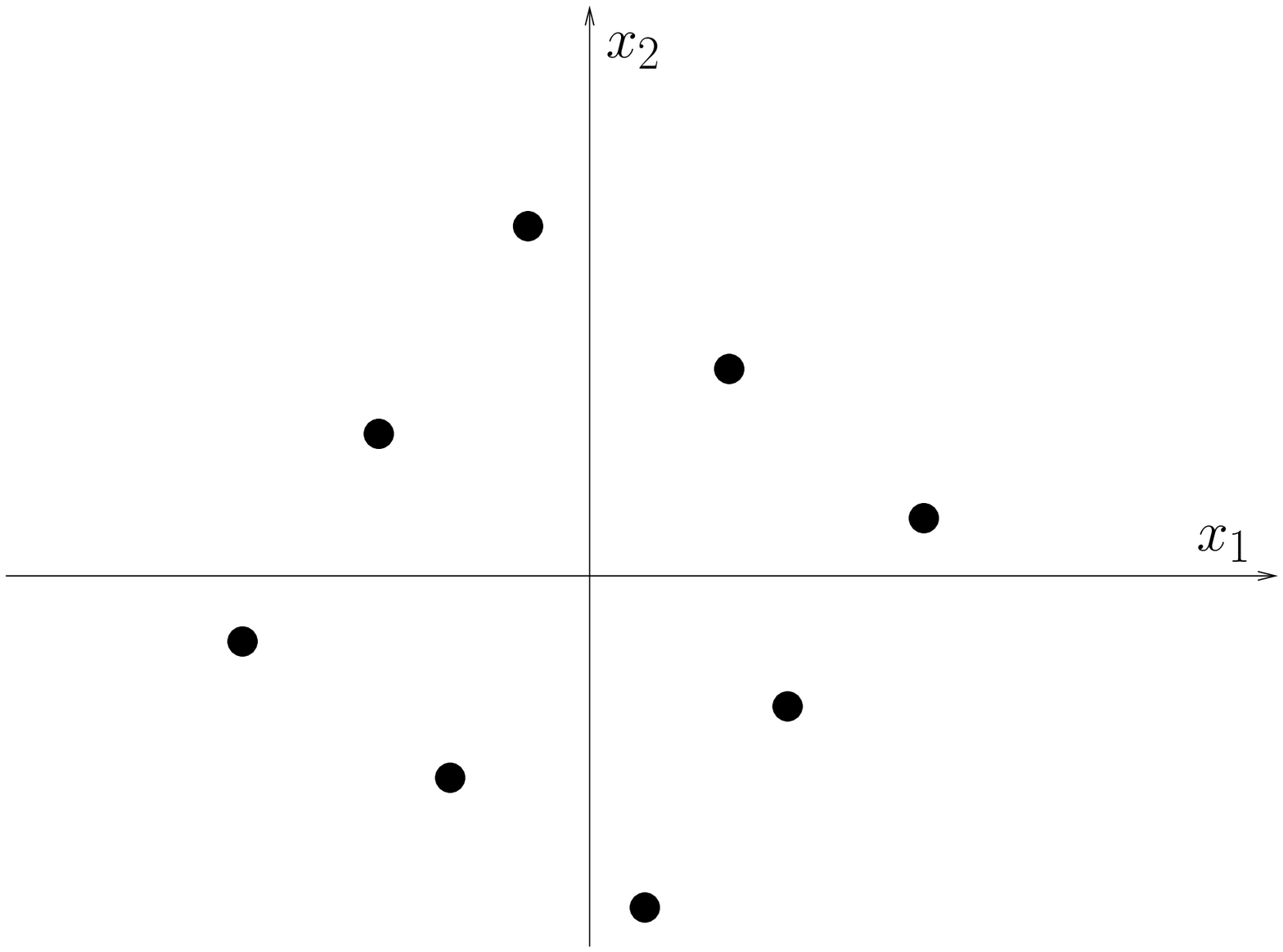}\label{fig: balanced 8-QAM}}
	\quad 
	 \subfigure[The received symbols $t_b$ are the faded components $\alpha_b x_b$. Collecting the different components $t_b$, the faded constellation $\Omega_t$ is constructed, which is expressed by $\Omega_t=\bs{\alpha}\cdot\Omega_x$.]{\includegraphics[width=0.40 \textwidth]{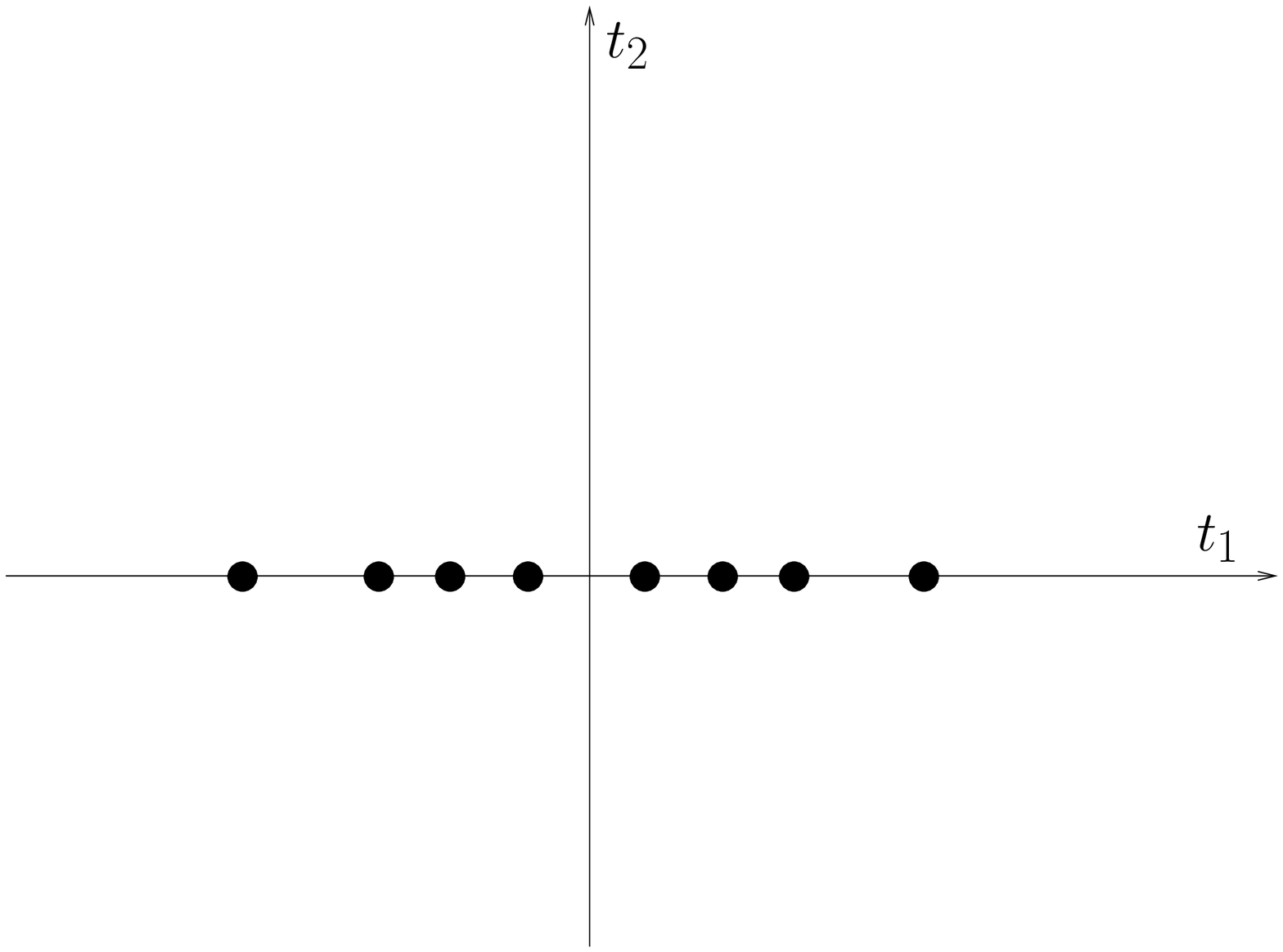}\label{fig: unbalanced 8-QAM}}
	\caption{The transmitted ($\Omega_x$) and faded constellation ($\Omega_t$) are shown for the case of the transmission of real-valued symbols, $B=2$ and $\Omega_z=8$-$\mathcal{R}^2$. For $\Omega_t = \bs{\alpha}\cdot\Omega_x$, the fading point $(\alpha_{1,\textrm{o}},0)$, the intersection between the outage boundary and the axis $\alpha_{2}=0$, is used.}
	\label{fig: rotated 8-QAM}
\end{figure}

Also, the outage region of block fading channels with  precoded constellations is inner bounded  by a $B$-hypersphere with center in the origin and touching the outage boundary  in the point $\bs{\alpha}_{\textrm{e}}$. Hence, $\alpha_{\textrm{e}}$ determines the radius of this $B$-hypersphere, $\sqrt{B} \alpha_{\textrm{e}}$. In this point, the faded constellation $\Omega_t$ is balanced, as in Fig. \ref{fig: balanced 8-QAM}. We will denote this balanced constellation $\Omega_t = \bs{\alpha}_{\textrm{e}} \cdot\Omega_x$ by $\mathcal{S}_{\textrm{e}}$. Hence, the hypersphere corresponds to a lower bound on the outage probability which is minimized by optimizing the mutual information of $\mathcal{S}_{\textrm{e}}$. 

In this section, we will show that the radius of the outer region, $\alpha_{\textrm{o}}$, and the radius of the inner region, $\sqrt{B} \alpha_{\textrm{e}}$, can be minimized by combining a simple optimization of the precoding matrix $P$ with a constellation expansion.

\subsection{Optimization of the precoding matrix}

Precoding corresponds to performing a unitary transformation of the input vector $\mb{x}$. As a unitary transformation preserves distance, it follows from Eq. (\ref{eq.: mutual info 2}) that the value $\alpha_{\textrm{e}}$ is insensitive to orthogonal transformations. Hence, the selection of the precoding matrix $P$ affects only the radius of the outer region $\alpha_{\textrm{o}}$. Let us denote by $\mathcal{O}$ the set of parameters from $P$ over which we will minimize $\alpha_{\textrm{o}}$. For $B=2$ and $B=3$, the only degree of freedom is the rotation angle (see (\ref{eq: P for B=2}) and (\ref{eq: P for B=3})). For $B>3$, more degrees of freedom can be exploited to minimize $\alpha_{\textrm{o}}$. For the numerical results, we restrict ourselves to $B \leq 3$. 

The mutual information of $\mathcal{S}_{\textrm{p}}$ can be rewritten as $I_{\mathcal{S}_{\textrm{p}}}(\alpha_{\textrm{o}}^2 \gamma, \mathcal{O})$, which yields $\alpha_{\textrm{o}}^2 = \frac{I_{\mathcal{S}_{\textrm{p}}}^{-1}(B R, \mathcal{O})}{\gamma}$. Changing the value of $\mathcal{O}$ (e.g. the rotation angle $\theta$ for $B=2$) will change the distances between the points in $\mathcal{S}_{\textrm{p}}$ and so change its mutual information. For a fixed spectral efficiency $R$ and fixed average SNR $\gamma$, minimizing the radius yields the optimization criterion
\[ \mathcal{O}_{\textrm{opt}} = \argmin{\mathcal{O}}~ I_{\mathcal{S}_{\textrm{p}}}^{-1}(B R, \mathcal{O}).\] 
The optimization is performed by means of a simulation, due to the lack of closed form expressions of the mutual information. Because the constellation is one-dimensional, the computational effort is minimal. We apply this optimization for different scenarios in Sec. \ref{sec: numerical results}. 

\subsection{Constellation expansion}
\label{sec: constellation expansion}

As the number of information bits per channel use is $R = m R_c /B$, there are different combinations of $m$ and $R_c$ yielding the same $R$. Taking into account that $R_c \leq 1$, the minimum value of $m$ equals $\lceil B R \rceil$, with a corresponding coding rate $R_c = \frac{B R}{\lceil B R \rceil}$. 

The number of points in the constellation is $|\Omega_t|$. Increasing the constellation size of $\Omega_z$ will render a constellation $\Omega_t$ with more points, both for $\mathcal{S}_{\textrm{e}}$ and $\mathcal{S}_{\textrm{p}}$. This higher order constellation may need less energy to achieve the same rate, both for the balanced case (optimization of $\alpha_{e}$) as the distorted case (optimization of $\alpha_{\textrm{o}}$). However, the decoding complexity increases as well as the complexity of optimization, so that there is a trade-off between performance and complexity. The higher the constellation size, the smaller the horizontal SNR-gap between the outage probabilities corresponding to a precoded discrete input alphabet and i.i.d. Gaussian input alphabet. However, the improvement in performance becomes smaller and smaller, as illustrated in Sec. \ref{sec: numerical results}. 


\section{Numerical results}
\label{sec: numerical results}

\subsection{Numerical results for $B=2$}

When $B=2$, $\mathcal{O}=\theta$, and the optimization criterion for the upper bound on the outage probability is to find $\theta$ so that $I_{\mathcal{S}_{\textrm{p}}}^{-1}(B R, \theta)$ is minimized. Next, a constellation expansion is performed to further minimize the upper bound as well as the lower bound on the outage probability. 

\subsubsection{Real constellations}

Assume that a transmission rate $R=0.9$ bpcu is aimed. First, we consider the optimization of the rotation angle $\theta$, see Fig. \ref{fig: optimization constellation expansion}. On the left $y$-axis, we show the instantaneous SNR per symbol, $\gamma_{\textrm{s}} = \alpha_{\textrm{o}}^2 \gamma$, so that $I_{\mathcal{S}_{\textrm{p}}}(\gamma_{\textrm{s}}, \theta) = B R$. The minimum SNR per symbol $\gamma_{\textrm{s}}$ that is needed to transmit $R=0.9$ bpcu for $\gamma=8$ dB is achieved by an i.i.d. Gaussian input alphabet:
\[\gamma_{\textrm{s}} = \frac{2^{4R}-1}{2}.\]
This fundamental minimum can be approached when using a precoded discrete input $\Omega_z=4$-$\mathcal{R}^2$ ($R_c = 0.9$) with rotation angle $\theta=27$ degrees. Now, we apply a constellation expansion to further reduce $\alpha_{\textrm{o}}$ (see Fig. \ref{fig: optimization constellation expansion}) and $\alpha_{\textrm{e}}$ (see Fig. \ref{fig:  outage   boundary  8-QAM}). For example, $\gamma_{\textrm{s}}$ for the rotated constellation $\Omega_z=8$-$\mathcal{R}^2$ ($R_c = 0.6$) approaches the theoretical minimum very closely for rotation angles within $[0 \ldots 9]$ degrees. An expansion to $\Omega_z=16$-$\mathcal{R}^2$ ($R_c = 0.45$ and $\theta_{\textrm{opt}} \in [35, 45]$ degrees) only slightly improves the performance. The optimization of $\gamma_{\textrm{s}}$ decreases the volume of the region $V_{\textrm{o}}$, which is illustrated in Fig. \ref{fig: example outage boundary} for $\Omega_z=4$-$\mathcal{R}^2$. Fig. \ref{fig: outage boundary 8-QAM} illustrates that constellation expansion is sufficient to reduce the value $\alpha_{\textrm{e}}$, which is very close to the theoretical minimum. Therefore, the constellation is not further shaped to minimize the lower bound of the outage probability.

\begin{figure}[!t]
\begin{center}
\includegraphics[angle=-90, width=0.75\textwidth]{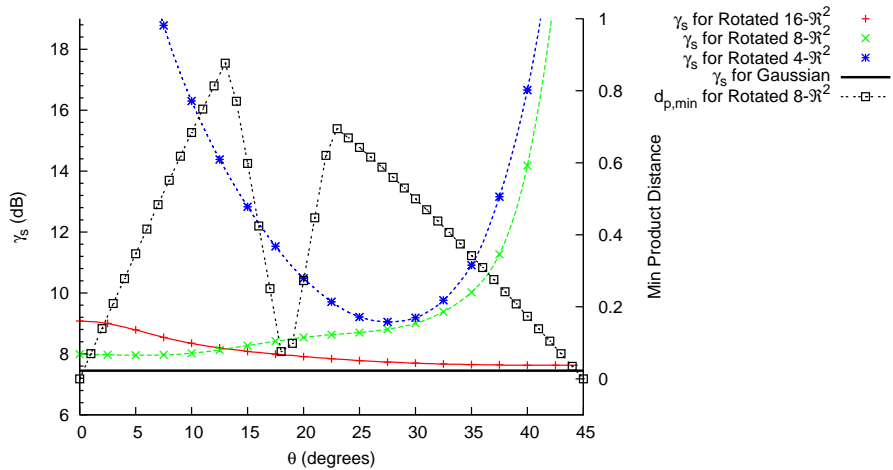}
\caption{The optimization of the radius of the outer region is shown for $\Omega_z=4$-$\mathcal{R}^2$, $R=0.9$ bpcu and $\gamma=8$ dB. The y-axis at the left denotes the instantaneous SNR per symbol, $\gamma_{\textrm{s}} = \alpha_{\textrm{o}}^2 \gamma$, and the right y-axis denotes the minimum product distance. The thick black line without markers represents the fundamental minimum SNR per symbol, $\gamma_{\textrm{s}}$, that is needed to transmit $R=0.9$ bpcu for $\gamma=8$ dB, i.e., when using an i.i.d. Gaussian input alphabet. The effect of constellation expansion on the radius $\alpha_{\textrm{o}}$ is shown by optimization of $\Omega_z=8$-$\mathcal{R}^2$ and $\Omega_z=16$-$\mathcal{R}^2$. The profile of the minimum product distance, $d_{p,\textrm{min}}$, is shown for $\Omega_z=8$-$\mathcal{R}^2$. \label{fig: optimization constellation expansion}}
\end{center}
\end{figure}

The information theoretic approach used in this paper does not lead to the same optimized rotations as in the case of algebraic constructions of uncoded constellations. In \cite{boutros1998ssd} and \cite{bayer2004nac}, multidimensional rotations have been optimized for uncoded infinite constellations transmitted on ergodic fading channels.  As a simple illustration, we show in Fig. \ref{fig: optimization constellation expansion} the minimum product distance $d_{p,\textrm{min}}$ \cite{tse2005fwc} of the uncoded $8$-$\mathcal{R}^2$ versus the rotation angle. The optimum and the profile of $d_{p,\textrm{min}}$ and those of $\gamma_{\textrm{s}}$ do not match. The minimum product distance approach is not suitable for coded schemes. For example, the minimum product distance is zero for $\theta=0$ degrees, while this rotation angle is optimal in terms of outage probability. 

\begin{figure}
	\centering
	\subfigure[]{\includegraphics[angle=-90, width=0.38 \textwidth]{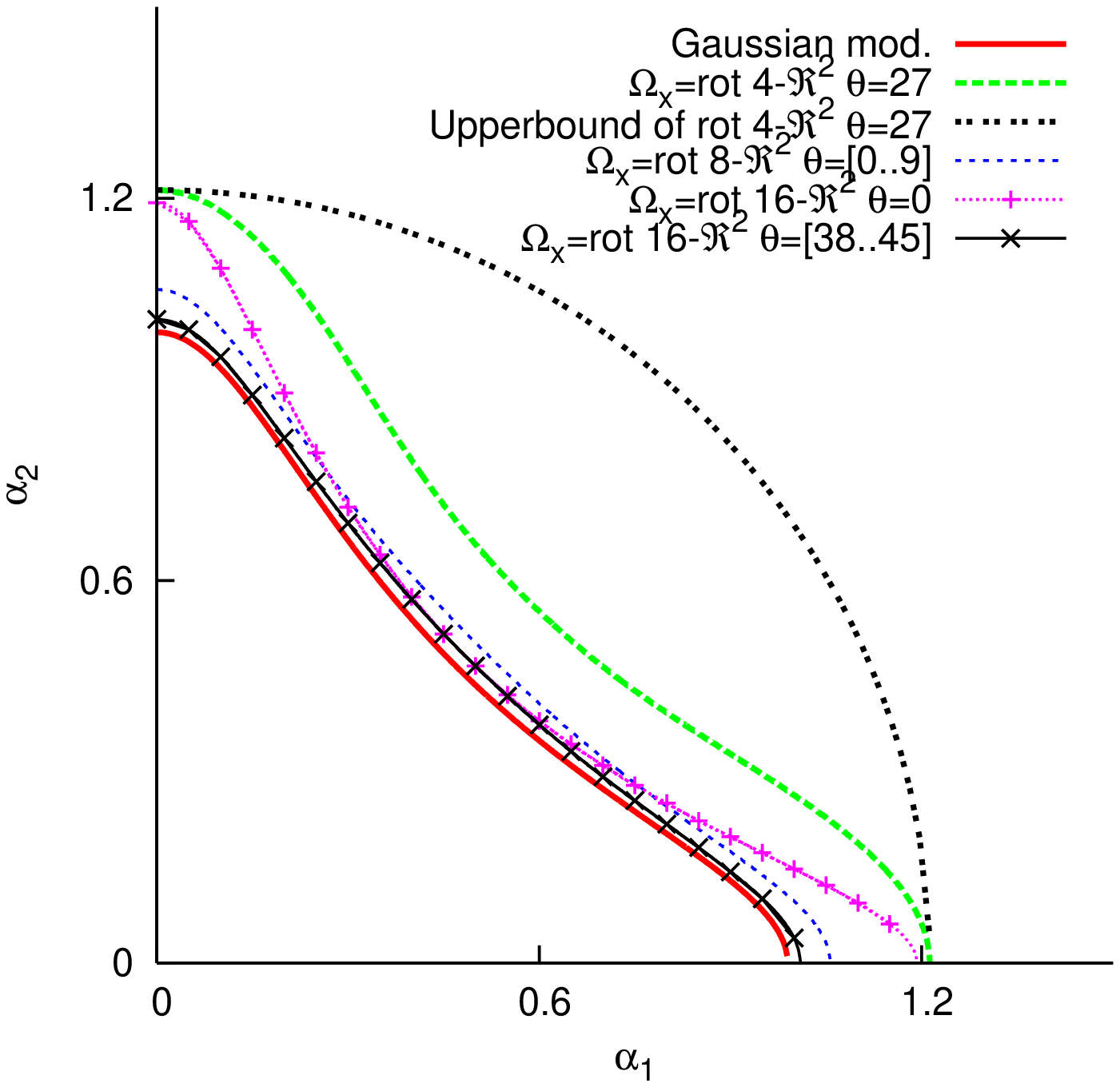}\label{fig: outage boundary 8-QAM}}
	\quad 
	 \subfigure[]{\includegraphics[angle=-90, width=0.5 \textwidth]{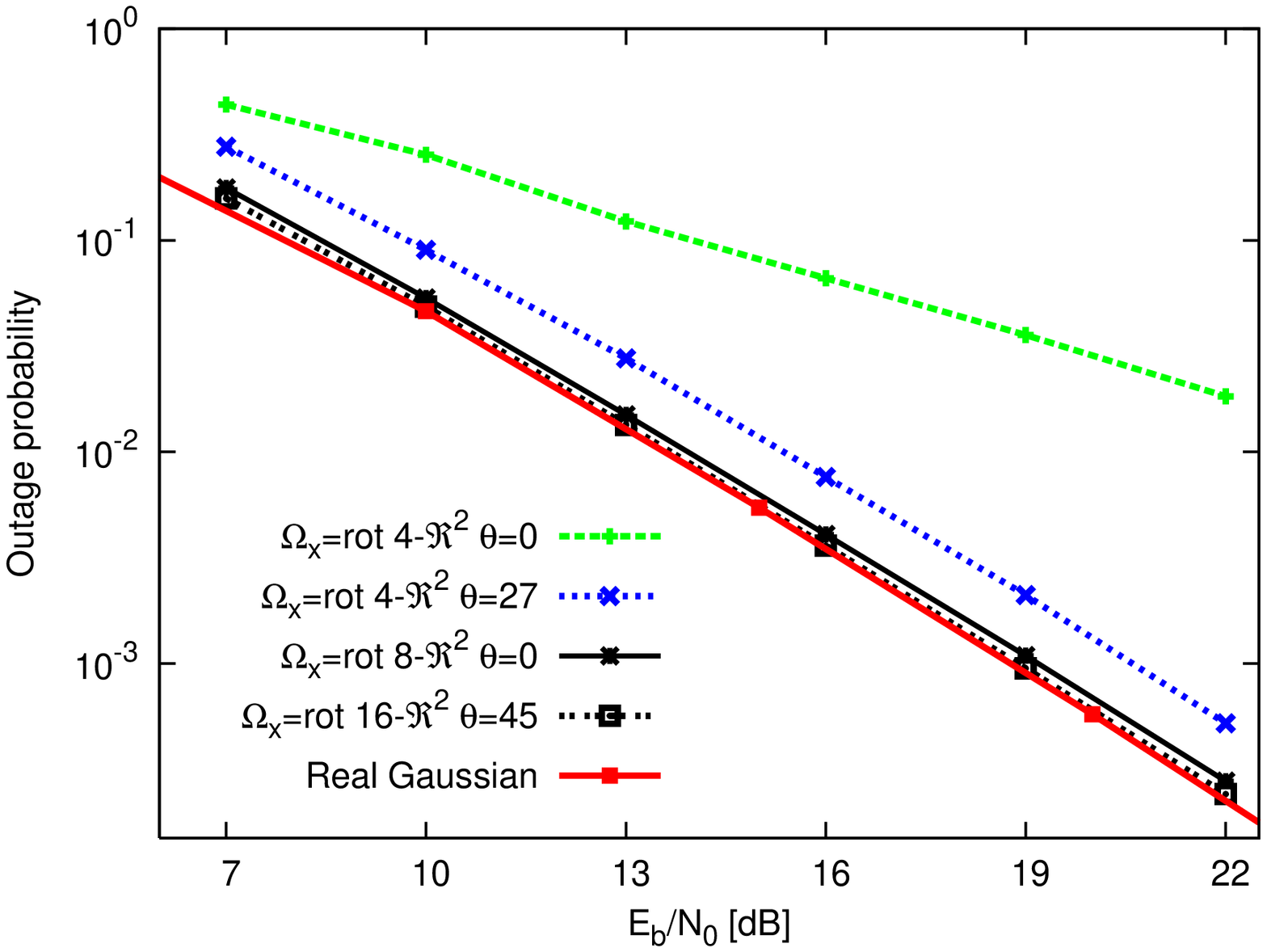}\label{fig: outage probability}}
	\caption{The outage boundaries (left) and outage probabilities (right) of $\Omega_z=4$-$\mathcal{R}^2$, $\Omega_z=8$-$\mathcal{R}^2$ and $\Omega_z=16$-$\mathcal{R}^2$ with and without optimized rotation angle are shown. The spectral efficiency is $R=0.9$ bpcu and $\gamma=8$ dB (left).}
	\label{fig: outage boundary and probability complex}
\end{figure}

The outage probabilities of the considered multidimensional constellations are shown in Fig. \ref{fig: outage probability}. This confirms that constellation expansion together with the optimization of the precoding parameter is sufficient to approach the outage probability with an i.i.d. Gaussian input alphabet very closely. It also shows that the constellation $\Omega_x = 8$-$\mathcal{R}^2$ with $\theta \in $ $[0 \ldots 9]$ degrees represents the best trade-off between performance and complexity.

\subsubsection{Extension to complex constellations}
\label{sec: optimization complex}

All the proofs in this paper are valid for complex constellations. This means that also for complex constellations, the outage region is outer bounded by a $B$-hypersphere, determined by one variable, its radius. We restrict our attention to real-valued precoding matrices. Consider Eq. (\ref{eq: precoding}), where $\mb{z}$ is now a complex vector and $P$ is real-valued. For complex symbols, this can be rewritten as 
\begin{equation}
\label{eq: complex rotation}
	\mb{x} = P \mathcal{R}\{ \mb{z} \} + j P \mathbb{I}\{ \mb{z} \},
\end{equation}
where $j^2 = -1$, $\mathcal{R}\{ .\}$ and $\mathbb{I}\{.\}$ take the real and complex part respectively. This means that the real and imaginary part of the complex vector are each precoded by the same matrix $P$. 

Assume that a transmission rate $R=1.8$ bpcu is aimed. Initially, we take the constellation $\Omega_z = 16$-$\mathcal{C}^2$ ($R_c = 0.9$), which can be build as the Cartesian product of two 4-QAM constellations ($16$-$\mathcal{C}^2$=4-QAM$\times$4-QAM). As for real-valued constellations, the rotation angle $\theta$ can be optimized, see Fig. \ref{fig: optim theta complex}. The gap to the outage probability corresponding to an i.i.d. Gaussian input alphabet can be closed by a constellation expansion and a new optimization of the rotation angle. The same strategy as for real-valued constellations could be applied by only adding one bit in the multidimensional constellation, which would extend $\Omega_z = 16$-$\mathcal{C}^2$ to $\Omega_z = 32$-$\mathcal{C}^2$. However, $\Omega_z = 32$-$\mathcal{C}^2$ cannot be written as the Cartesian product of two constellations and is therefore less convenient to generate ($32$ points would have to be placed properly in a $4$-dimensional space). For simplicity, the constellation expansion is done by adding one bit per component which extends $\Omega_z = 16$-$\mathcal{C}^2$ (=4-QAM$\times$4-QAM, $R_c = 0.9$) to $\Omega_z = 64$-$\mathcal{C}^2$ (8-QAM$\times$8-QAM \footnote{The 8-QAM constellations have the same form as in Fig. \ref{fig: balanced 8-QAM}.}, $R_c = 0.6$). The optimization of $\theta$ and the optimized outage probabilities are shown in Fig. \ref{fig: outage rotated complex}.

\begin{figure}
	\centering
	\subfigure[The optimization of the rotation angle for complex symbols is shown.]{\includegraphics[angle=-90, width=0.40 \textwidth]{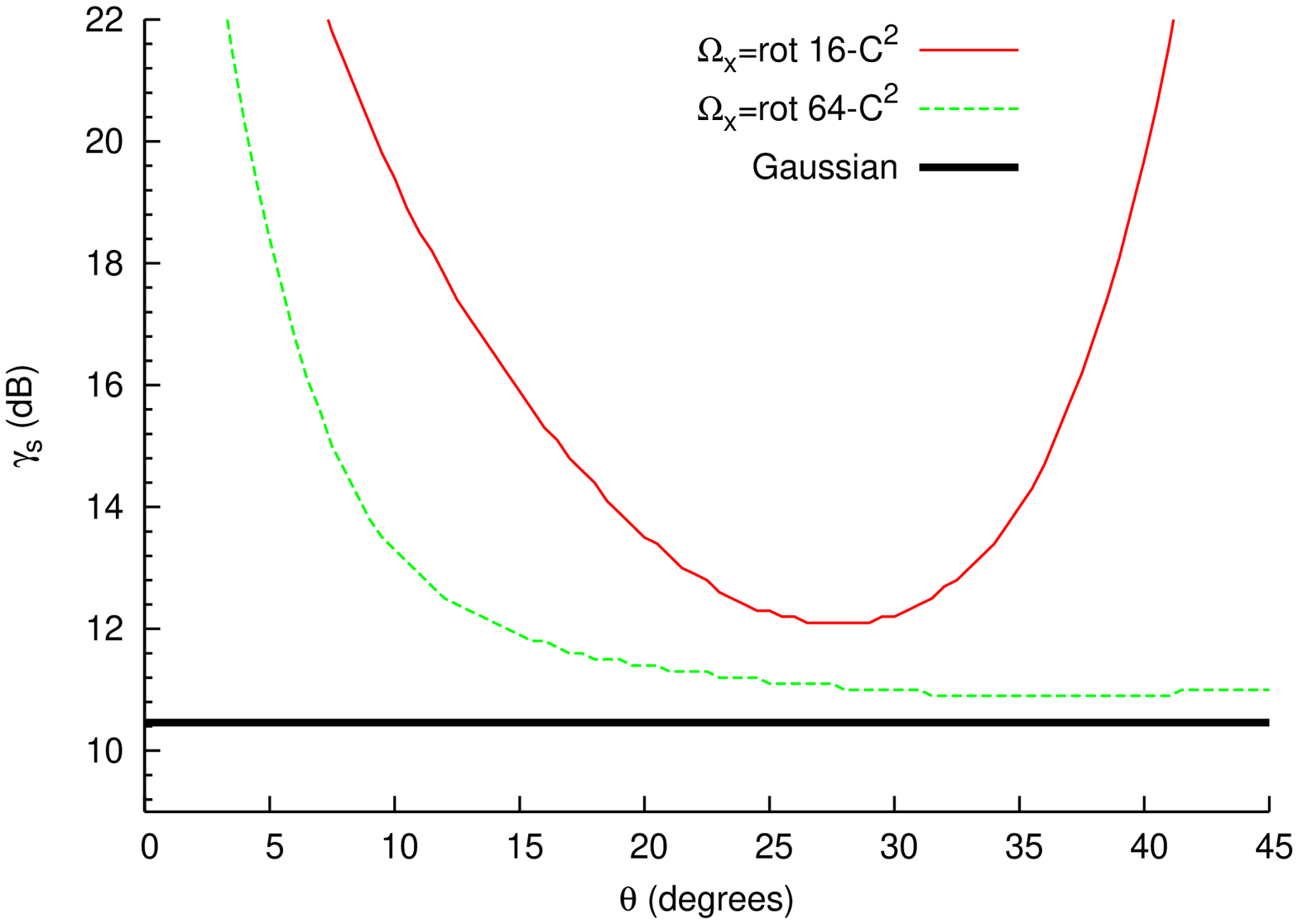}\label{fig: optim theta complex}}
	\quad 
	 \subfigure[The outage probabilities of the BF channel with complex inputs are shown.]{\includegraphics[angle=-90, width=0.40 \textwidth]{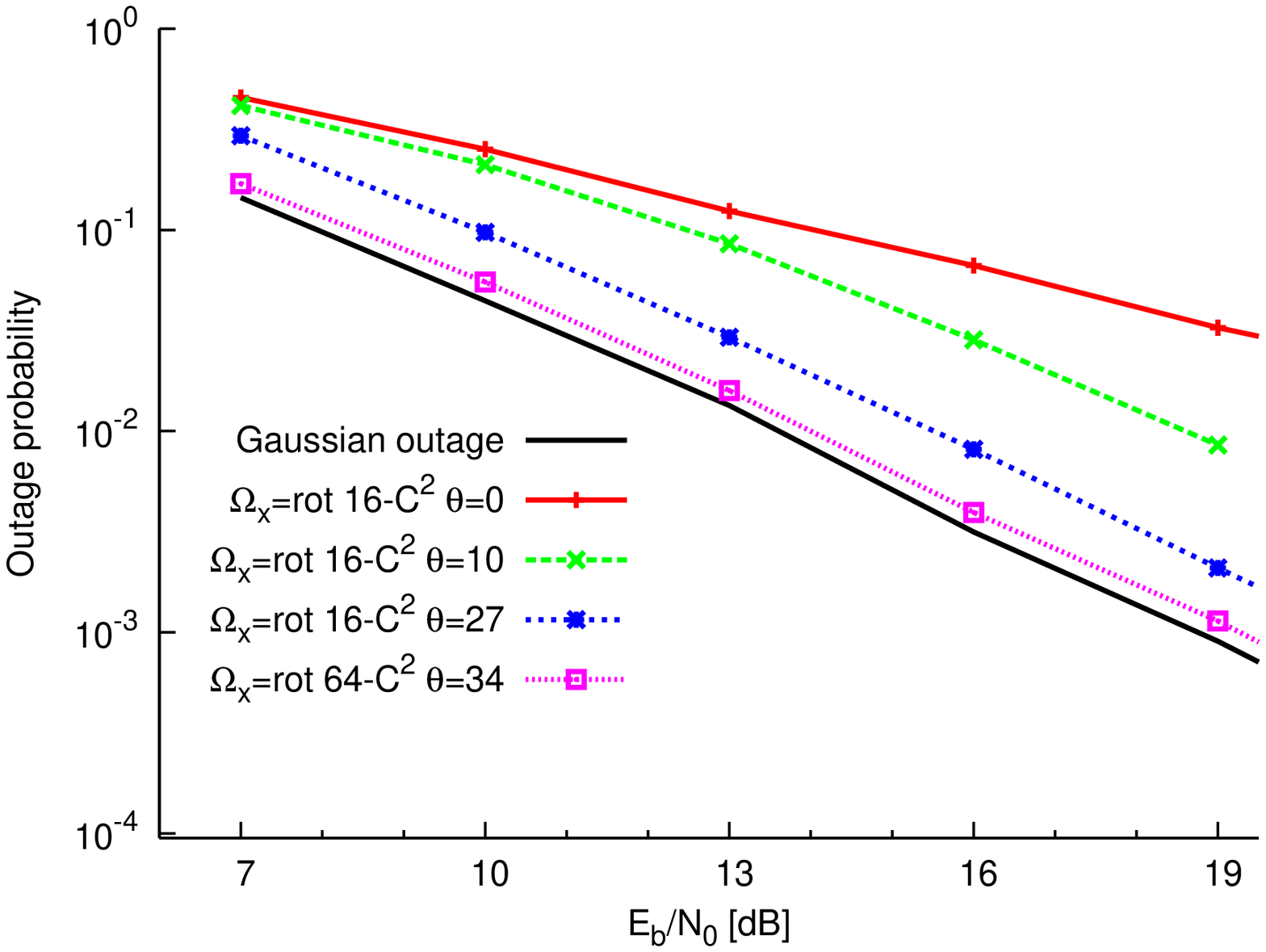}\label{fig: outage complex}}
	\caption{The optimization of $\theta$ and the optimized outage probabilities are shown when complex symbols are transmitted. The transmitted rate is $R=1.8$ bpcu.}
	\label{fig: outage rotated complex}
\end{figure}

Note that, for $\Omega_z = 16$-$\mathcal{C}^2$, the profile of the rotation angle $\theta$, and so the optimum rotation angle, is the same as for $\Omega_z = 4$-$\mathcal{R}^2$ and $R=0.9$ bpcu. This can be explained as follows. When $\mathcal{R}\{ \mb{z} \}$ and $\mathbb{I}\{ \mb{z} \}$ are drawn from the same real-valued constellation $\Psi_z$, then $\mb{z}$ belongs to a constellation $\Omega_z = \Psi_z + j \Psi_z$. Consequently, $\mb{x}$ belongs to a constellation $\Omega_x = \Psi_x + j \Psi_x$, where $\Psi_x$ is obtained by applying the precoding matrix $P$ to the constellation $\Psi_z$. From the chain rule of mutual information \cite{cover2006eit}, we obtain
\begin{equation}
\label{eq: complex QPSK}
	I_{\Omega_x} \left(\gamma \right) = 2 I_{\Psi_x}\left(\frac{\gamma}{2} \right).
\end{equation}
Hence, the precoding matrix $P$ that is optimum for a real-valued constellation $\Psi_z$ and rate $R$ is also optimum for a complex constellation $\Omega_z = \Psi_z + j \Psi_z$ and rate $2R$. The corresponding SNR for the complex constellation is $3$ dB higher than for the real-valued constellation. For $\Omega_z = 16$-$\mathcal{C}^2$ and $R = 1.8$ bpcu, the corresponding real-valued constellation is $\Psi_z = 4$-$\mathcal{R}^2$. Therefore, the profiles in Figs. \ref{fig: optim theta complex} and \ref{fig: optimization constellation expansion} are the same for both constellations, except for an upward translation of $3$ dB. \\

We also tested the performance for higher spectral efficiencies. Consider for example a system requiring a spectral efficiency of $R=3.6$ bpcu. Here, the same techniques can be used. First, $\Omega_z = 256$-$\mathcal{C}^2$ (16-QAM$\times$16-QAM, $R_c = 0.9$) is optimized, followed by $\Omega_z = 1024$-$\mathcal{C}^2$ (32-QAM$\times$32-QAM \footnote{The well known cross 32-QAM constellations are used.}, $R_c = 0.72$). The results are given in Fig. \ref{fig: outage probability high rate}. The same observations hold as for the previous numerical results.

\begin{figure}[t!]
	\centering
	\includegraphics[angle=-90, width=0.6\columnwidth]{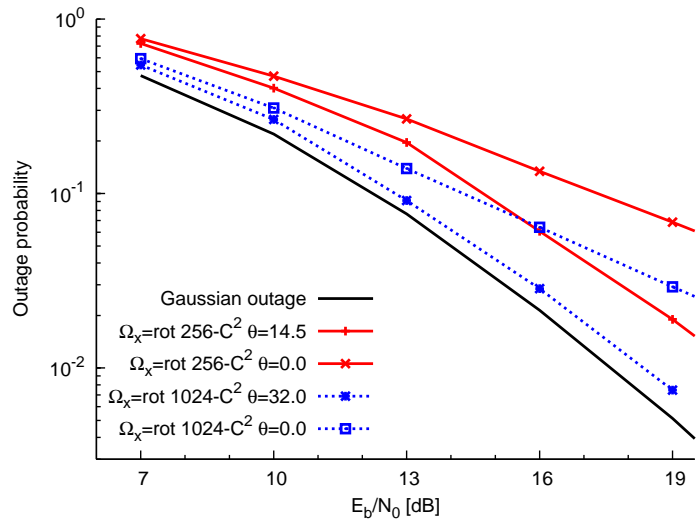}
	\caption{Outage probabilities of the BF channel with input $\Omega_z = 256$-$\mathcal{C}^2$ and $\Omega_z = 1024$-$\mathcal{C}^2$. The spectral efficiency is $R=3.6$ bpcu.}
	\label{fig: outage probability high rate}
\end{figure}

\subsection{Numerical results for $B=3$}
\label{Sec: mfg}

When $B=3$, $\mathcal{O}=\theta_1$ (see Eq. (\ref{eq: P for B=3})), and the optimization criterion for the upper bound on the outage probability is to find $\theta_1$ so that $I_{\mathcal{S}_{\textrm{p}}}^{-1}(\theta_1,B R)$ is minimized. Next, a constellation expansion is performed to further minimize the upper bound as well as the lower bound on the outage probability. 

We aim to transmit $R=0.9$ bpcu. We construct constellations $\Omega_z$ that are invariant to a cyclic shift of the components of the constellation points. Hence, the constellations are invariant to a rotation of $2 \pi/3$ with respect to the bisector $(1,1,1)$ (this can be verified by evaluating $P$ in Eq. (\ref{eq: P for B=3}) for this rotation angle). An example of such constellations are the Cartesian products of three identical one-dimensional constellations, such as $8$-$\mathcal{R}^3 = $(BPSK$)^3$ (constellation points are corners of a cube) and $64$-$\mathcal{R}^3 = (4$-PAM$)^3$ (constellation points are the corners of four nested cubes). More sophisticated constellations can also be considered. For example, $16$-$\mathcal{R}^3$ can be constructed by considering 5 sets of three constellation points $\{ (a_j, b_j, c_j), (b_j, c_j, a_j), (c_j, a_j, b_j)  \}$ \footnote{The three constellation points are in a plane perpendicular to the bisector.} for $j=1, \ldots, \ldots, 5$, and adding a constellation point located on the bisector. To build the actual constellations, a few design parameters have to be specified, such as the distances between the planes perpendicular to the bisector, the radii of the circles containing three points of the constellation in the planes, how many points on the bisector are taken and how many points in groups of three are taken. These design parameters will impact on the final performance. 

Because it is not the main topic of the paper, we will not elaborate on many different designs for the multidimensional constellations. We compare the performance of $8$-$\mathcal{R}^3$ ($R_c = 0.9$), $16$-$\mathcal{R}^3$ ($R_c = 0.675$) and $64$-$\mathcal{R}^3$ ($R_c = 0.45$) in Fig. \ref{fig: outage rotated mfg}. 
\begin{figure}
	\centering
	\subfigure[The optimization of the rotation angle for $B=3$ is shown.]{\includegraphics[angle=-90, width=0.40 \textwidth]{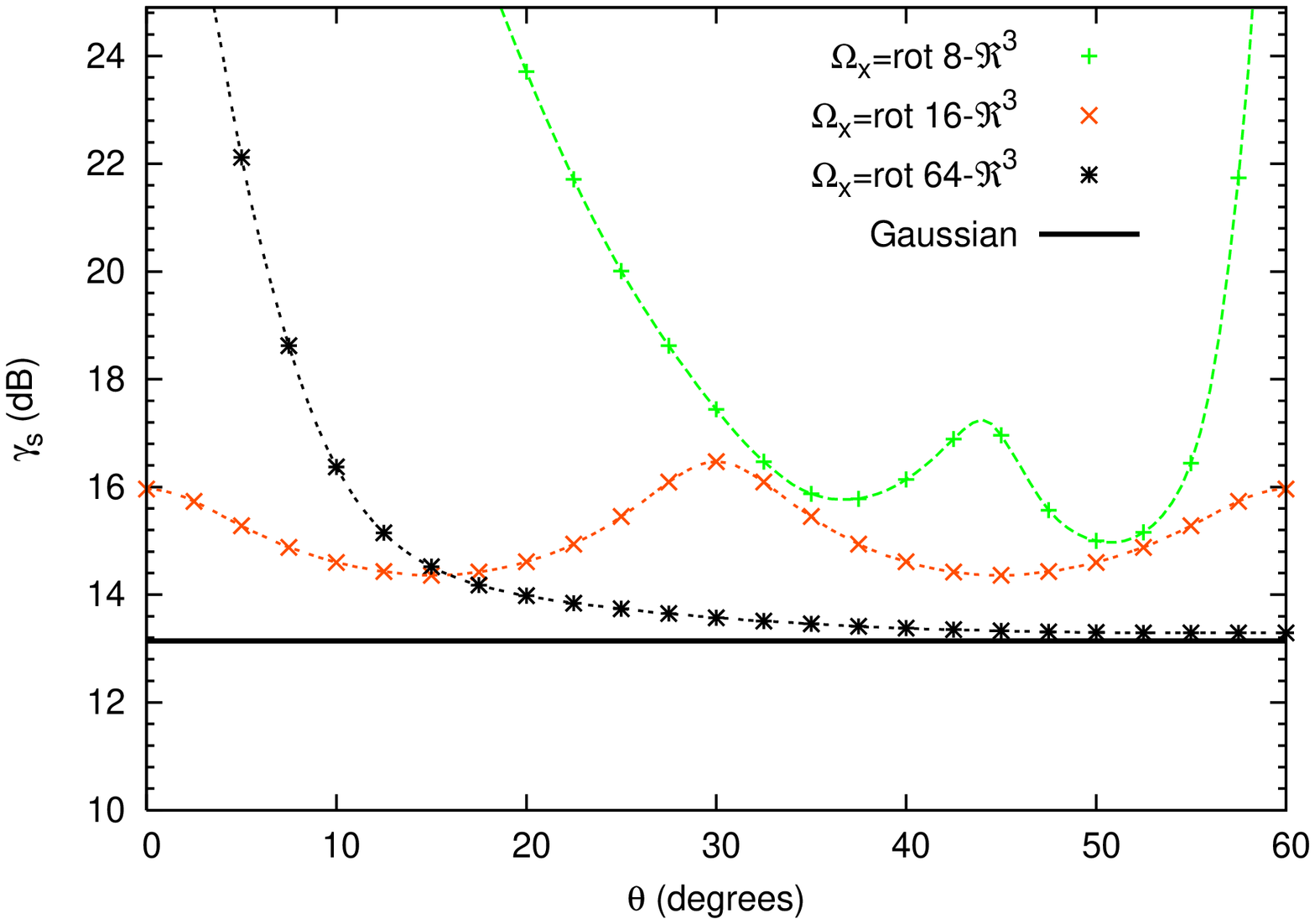}\label{fig: optim theta mfg}}
	\quad 
	 \subfigure[The outage probabilities of the BF channel for $B=3$ are shown.]{\includegraphics[angle=-90, width=0.40 \textwidth]{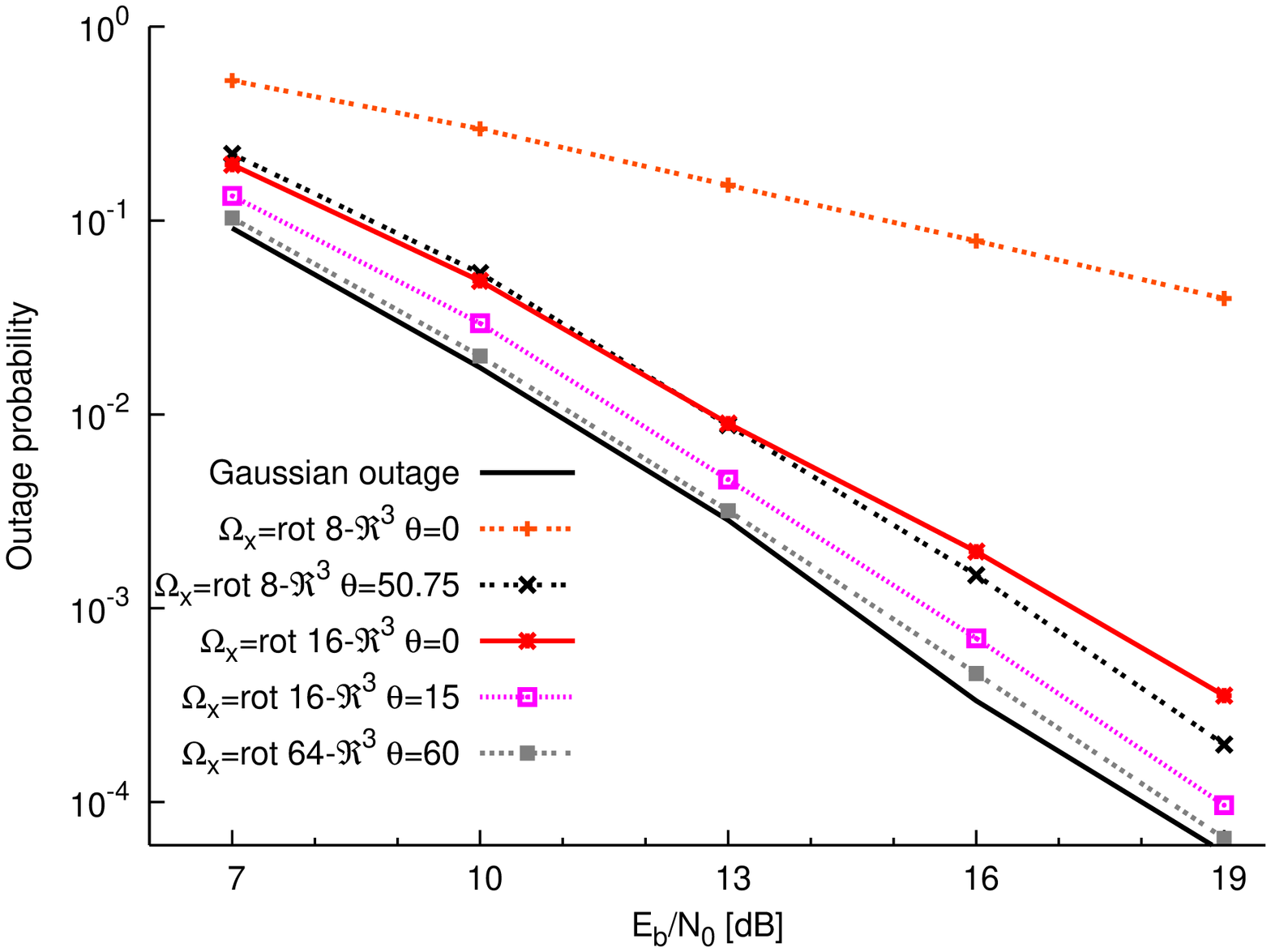}\label{fig: outage mfg}}
	\caption{The optimization of $\theta_1$ and the optimized outage probabilities are shown for $B=3$. The transmitted rate is $R=0.9$ bpcu.}
	\label{fig: outage rotated mfg}
\end{figure}
Note that the symmetry point in Fig. \ref{fig: optim theta mfg} is not $\theta_1=45$ degrees, as for $B=2$, but it is $60$ degrees. The $8$-$\mathcal{R}^3$ constellation is equal to the Cartesian product of three BPSK constellations and the $64$-$\mathcal{R}^3$ constellation is the Cartesian product of three 4-PAM constellations. For the $16$-$\mathcal{R}^3$ constellation, we chose to take a circle of three points in a plane containing the origin perpendicular to the bisector, and two circles with 6 points, each in a plane next to the first plane, perpendicular to the bisector. Finally the origin is also chosen as a constellation point. The rounded coordinates of the points of the constellation $\Omega_z$ are $\{ (0.0, 1.0, -1.0), (1.9, 0.12, 0.12), (1.3, -0.5, 1.3), (0.5 -1.3 -1.3), (-1.9, -0.1, -0.1), (0,0,0) \}$  as well as the cyclic shifts of these coordinates.

\subsection{Practical implications}
\label{subsec: Practical implications}

The motivation for this work was to solve the theoretical problem of determining the optimal precoding matrix $P$ in terms of the outage probability for a given combination $\{R_c, \Omega_z\}$ satisfying a target spectral efficiency $R$. Besides the theoretical relevance, the numerical results suggest a practical relevance because the impact of this optimization on the outage probability is significant in the case of large coding rates $R_c$ \footnote{Of course, when the coding rate decreases, the impact of the error-correcting code on the error rate performance increases yielding a decreased impact of the precoder on the latter.}. This case might be preferred upon the case of small coding rates (and thus large constellation sizes for a given spectral efficiency $R$) because large constellation sizes complicates the system design (requiring a joint optimization of the constellation labeling and error-correcting code which is not trivial) and is sometimes even not allowed (e.g., $|\Omega_z| \leq 64$ in $4$G systems).

The consequences of our work on practical code design are as follows. The WER of a practical system is lower bounded by the outage probability (the lower bound is achievable). To minimize the WER, first its lower bound must be minimized. Therefore, the multidimensional constellation $\Omega_z$ and the rotation angle interval for the practical code should be taken as obtained in this work. Next, the labelling, the rotation angle within the rotation angle interval obtained in this paper, and the error-correcting code must be determined. The last three optimizations are the topics of another work \cite{duy2010cmf}, but it is important to understand that these optimizations are based on what is presented in this paper. In Fig. \ref{fig: results code}, we show for $B=2$ and $R=0.9$ bpcu that the optimized outage probability can be approached very closely by the WER of a practical system. 

\begin{figure}
	\centering
	\includegraphics[angle=-90, width=0.6 \textwidth]{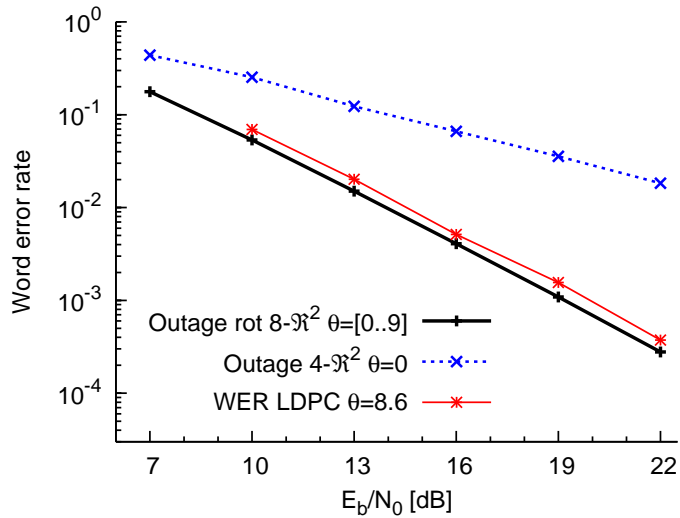}
	\caption{The optimized outage probability can be approached very closely by the WER of a practical error-correcting code. The block length is $N=5000$, but the results are valid for all block lengths (tested for $N=2000$ and $N=10000$).}
	\label{fig: results code}
\end{figure}

\section{Conclusions}

We have studied  the effect of linear precoding on  the outage probability of block  fading channels. We have analyzed the outage  boundaries in the  fading space  and established outer and inner boundaries with simple shapes which yield an easy  optimization of the outage probability for a discrete constellation, for an arbitrary number of blocks in the fading channel, real-valued or complex constellations, low or high spectral efficiency. The combination of a constellation expansion and an optimized precoding matrix, has shown to be sufficient to closely approach the outage probability corresponding to an i.i.d. Gaussian input alphabet. With this work, the practical code performance can be optimized by admitting the parameters obtained here. 

\section*{Acknowledgement}

The work of Joseph Boutros and part of the work of Dieter Duyck were supported by the Broadband Communications Systems project funded by Qatar Telecom (Qtel). Dieter Duyck and Marc Moeneclaey wish to acknowledge the activity of the Network of Excellence in Wireless COMmunications NEWCOM++ of the European Commission (contract n. 216715) that motivated this work.

\appendix

\subsection{Symmetry conditions for multidimensional constellations}
\label{appendix: symmetry}

The points $\{\bs{\alpha}_{b,\textrm{o}}, b=1,\ldots, B \}$ correspond to the case that all fading gains are zero, except one, whose value is the scaling factor of the projection of the multidimensional constellation on the b-th coordinate axis $x_b$, so that the mutual information between $\mb{X}$ and $\mb{Y}$ is equal to the spectral efficiency $B R$. In other words, if the projection of the multidimensional constellation on each coordinate axis yields the same set of points, then the magnitudes of the points $\{ \bs{\alpha}_{b,\textrm{o}}, b=1,\ldots, B \}$ are equal.

First, we restrict our attention to the case that $B=2$. Consider the constellation point $\mb{z}^{(i)} = (u^{(i)}_1, u^{(i)}_2) \in \Omega_z$. The projection of the multidimensional constellation on each coordinate axis yields the same set of points if for each point $\mb{z}^{(i)}$, the points $\mb{z}^{(j)} = (u^{(j)}_1, u^{(j)}_2)$ and $\mb{z}^{(q)} = (u^{(q)}_1, u^{(q)}_2)$ exist, $i,j,q \in [1, \ldots, 2^m]; ~j,q \neq i$, so that 
\begin{equation*}
\left\{
\begin{array}{c}
\cos(\theta)u^{(i)}_1 - \sin(\theta) u^{(i)}_2  = \sin(\theta) u^{(j)}_1 + \cos(\theta) u^{(j)}_2 \\
\sin(\theta) u^{(i)}_1 + \cos(\theta)u^{(i)}_2   = \cos(\theta) u^{(q)}_1 -  \sin(\theta) u^{(q)}_2 . 
\end{array}
\right. 
\end{equation*}
In other words, $x^{(i)}_1 = x^{(j)}_2$ and $x^{(i)}_2 = x^{(q)}_1$, where $\mb{x}^{(i)}, \mb{x}^{(j)}$ and $\mb{x}^{(q)}$ are the corresponding points of $\mb{z}^{(i)}, \mb{z}^{(j)}$ and $\mb{z}^{(q)}$ in $\Omega_x$. It can be easily verified that this is always fulfilled if 
\begin{equation*}
\left\{
\begin{array}{c}
(u^{(i)}_1,u^{(i)}_2) = (u^{(j)}_2,-u^{(j)}_1) \\
(u^{(i)}_1,u^{(i)}_2) = (-u^{(q)}_2,u^{(q)}_1),
\end{array}
\right.
\end{equation*}
or in other words, the constellation is invariant under a rotation of $\pi/2$, which is obtained after a cyclic shift and a reflection, which proves what was claimed. An example of such a constellation is the constellation shown in Fig. \ref{fig: balanced}.

Now consider the case that $B>2$. Consider the $B$-dimensional constellation $\Omega_z$ that contains $M$ points. When $\mb{z}$ belongs to $\Omega_z$, then also $\mb{z}^{(1)}, \ldots, \mb{z}^{(B-1)}$ belong to $\Omega_z$, where $\mb{z}^{(b)}$ is obtained from $\mb{z}$ by a $b$-fold upward cyclic shift of the components of $\mb{z}$: $\mb{z}^{(b)} = C^b \mb{z}$, where $C$ is obtained as a cyclic shift to the right of the columns of the $B \times B$ identity matrix. Note that the number of constellation points does not need to be a multiple of $B$: a subset of the constellation may consist of an arbitrary number of constellation points of the type $[z, z, \ldots, z]^T$ which remain invariant under a cyclic shift.

Consider an orthogonal circulant $B \times B$ precoding matrix $P$. Therefore, $P = C P C^T$ (a circulant matrix remains the same when applying a left cyclic shift to the columns and an upward cyclic shift to the rows). The transformation of $\mb{z}^{(b)}$ is
\begin{equation}
	P \mb{z}^{(b)} = P C^b \mb{z} = (C P C^T)C^b \mb{z} = C P C^{b-1} \mb{z} = \ldots = C^b P \mb{z} = C^b \mb{x} = \mb{x}^{(b)},
\end{equation}
where we exploit that $C$ is an orthogonal matrix. 

Consider the matrix $\left(\mb{x}, \mb{x}^{(1)}, \ldots, \mb{x}^{(B-1)} \right)$. As the $(i+1)$-th row is obtained as a cyclic shift to the left of the $i$-th row, the set of components in a row is the same for each row. A constellation point in $\Omega_z$ of the type $(z, z, \ldots, z)^T$ is transformed into a constellation point in $\Omega_x$ of the type $(x, x, \ldots, x)^T$. We conclude that the projection of the constellation $\Omega_x$ on any of the coordinate axes yields the same set of points.

\subsection{Proof of Lemma \ref{Lemma: Cartesian}}
\label{app: lemma}

The mutual information $I\left(\bs{\alpha},  \gamma,  P=I\right)$ is equal to $\frac{1}{B} I(\mb{Y};\mb{X}|\bs{\alpha}, \gamma)$. This can be split in the sum of $B$ terms through the chain rule \cite{cover2006eit}:
\begin{equation*}
	I\left(\bs{\alpha},  \gamma,  P=I\right)  = \frac{1}{B} \left( I\left(Y_1;\mb{X} | \alpha_1,  \gamma\right) + I\left(Y_2;\mb{X} | Y_1, \alpha_1, \alpha_2,  \gamma\right) + \ldots + I\left(Y_B;\mb{X} | Y_1, \ldots, Y_{B-1}, \bs{\alpha},  \gamma\right) \right).  
\end{equation*}
The mutual information 
\begin{align*}
	I\left( Y_b;\mb{X} | Y_1, \ldots, Y_{b-1}, \bs{\alpha},  \gamma \right) &= H\left( Y_b| Y_1, \ldots, Y_{b-1}, \bs{\alpha},  \gamma \right)  - H\left( Y_b|\mb{X}, Y_1, \ldots, Y_{b-1}, \bs{\alpha},  \gamma \right) \\
	&= H\left( Y_b|Y_1, \ldots, Y_{b-1}, \bs{\alpha},  \gamma \right) - H\left( Y_b|X_b, \mb{X}_{\sim b}, Y_1, \ldots, Y_{b-1}, \bs{\alpha},  \gamma\right) \\
	&= H\left( Y_b|Y_1, \ldots, Y_{b-1}, \bs{\alpha},  \gamma \right) - H\left( Y_b|X_b,Y_1, \ldots, Y_{b-1}, \bs{\alpha},  \gamma \right) \\
	&= I\left( Y_b ; X_b | Y_1, \ldots, Y_{b-1}, \bs{\alpha},  \gamma \right),
\end{align*}
where $\mb{X}_{\sim b} = [X_1, \ldots, X_{b-1}, X_{b+1}, \ldots, X_B]$. 
Because 
\begin{equation}
	H\left(Y_b | Y_1, \ldots, Y_{b-1}, \bs{\alpha},  \gamma\right) \leq  H\left(Y_b|\bs{\alpha},  \gamma \right)
\end{equation}
and 
\begin{equation}
	H\left(Y_b | X_b , Y_1, \ldots, Y_{b-1}, \bs{\alpha},  \gamma\right) = H\left(Y_b | X_b, \bs{\alpha},  \gamma\right), 
\end{equation}
it is clear that 
\begin{equation}
	I\left(X_b;Y_b | Y_1, \ldots, Y_{b-1}, \bs{\alpha},  \gamma\right) \leq I\left(X_b;Y_b | \alpha_b,  \gamma\right),
\end{equation}
where the upper bound equals $I_{\mathcal{S}_{\textrm{p}}}\left(\alpha_b^2 \gamma, P\right)$. 

\subsection{Proof of Props. \ref{bound: gauss} and \ref{bound: bpsk}}
\label{app: Proof of Prop. bound: gauss}

Consider a function $F(x)$ that is concave for $x \geq 0$. Hence, for arbitrary $L$, 
\begin{equation}
\label{eq: Jens ineq}
	F \left( \sum_{l=1}^L p_l \gamma_l \right) \geq \sum_{l=1}^L p_l F(\gamma_l),
\end{equation}
for $\sum_{l=1}^L p_l = 1, ~ p_l \geq 0$ and $\gamma_l \geq 0$ for $l=1, \ldots, L$, by Jensen's inequality. In addition, assume that $F(0)=0$. 

We construct the function $\sum_{b=1}^B F(\gamma_b)$ where $\sum_{b=1}^B \gamma_b=C$, $\gamma_b \geq 0$ and $\gamma_b = \gamma \alpha_b^2$  for $b=1,\ldots,B$. Hence, $\bs{\alpha}$ is on the surface of a hypersphere with squared radius $C/\gamma$.

From Eq. (\ref{eq: Jens ineq}) with $L = B$, we obtain
\[ \sum_{b=1}^B F(\gamma_b) = B \left( \frac{1}{B} \sum_{b=1}^B F(\gamma_b) \right) \leq B F\left( \frac{1}{B} \sum_{b=1}^B \gamma_b \right) = B .F\left(\frac{C}{B} \right) \]
The maximum value $B . F\left(\frac{C}{B} \right)$ is achieved for $\gamma_b = \frac{C}{B},~ b = 1,\ldots, B$.

Further, using Eq. (\ref{eq: Jens ineq}) with $L=2$,
\[ F (\gamma_b) = F\left( \frac{\gamma_b}{C} C + \frac{C-\gamma_b}{C}~ 0  \right) \geq \frac{\gamma_b}{C} F\left( C \right) + \frac{C-\gamma_b}{C} F\left( 0  \right) =  \frac{\gamma_b}{C} F\left( C \right)  \]
Summing over $b$ yields
\[ \sum_{b=1}^B F (\gamma_b) \geq \frac{F(C)}{C}  \sum_{b=1}^B \gamma_b = F(C). \]
The minimum value $F(C)$ is achieved when, for given $l \in \{1, \ldots, B  \}, ~ \gamma_l=C$ and $\gamma_b=0,$ for $b \neq l$.

\subsection{Proof of Prop. \ref{bound: bpsk with precoding low SNR}}
\label{app: proof prop. precoding low SNR}

The instantaneous SNR is $\frac{E_s}{N_0}=\frac{\mathbb{E}[|\mb{t}^2|]}{\mathbb{E}[|\mb{w}^2|]} = \gamma |\bs{\alpha}|^2$. In this proof, we assume real constellations for simplicity. The proof can be straightforwardly extended to complex extensions. We denote the real part of the complex noise $\mb{w}$ by $\mb{w^\prime}$.

Consider the mutual information $I\left(\bs{\alpha}, \gamma, P\right)$ of the constellation $\Omega_t$, given the fading gains $\bs{\alpha}$ (Eq. (\ref{eq.: mutual info})). This expression can be reformulated in terms of $\gamma$:
\begin{equation*}
	I\left(\bs{\alpha}, \gamma, P\right) =  \frac{m}{B} - \frac{2^{-m}}{B} \sum_{\mb{x} \in \Omega_x}  \E_{\mb{y}|\mb{x}} \left[ \Log_2 \left( \sum_{\mb{x}^{\prime} \in \Omega_x} \exp\left[ \gamma \left( d^2(\mb{y},\bs{\alpha}\cdot\mb{x}) - d^2(\mb{y}, \bs{\alpha}\cdot\mb{x^{\prime}}) \right) \right] \right) \right],
\end{equation*}
where $\E_{\mb{y}|\mb{x}}$ can be replaced by an expectation over the noise, $\E_{\mb{w^\prime}}$, ${w_b^\prime} \sim \mathcal{N}(0,1/(2\gamma))$. The argument of the exponential functions can be simplified, so that
\begin{equation*}
I\left(\bs{\alpha}, \gamma, P \right) =  \frac{m}{B} - \frac{2^{-m}}{B} \sum_{\mb{x} \in \Omega_x} \E_{\mb{w^\prime}} \left[ \Log_2 \left( \sum_{\mb{x}^{\prime} \in \Omega_x} \exp\left[ -\gamma d^2(\bs{\alpha}\cdot\mb{x}, \bs{\alpha}\cdot\mb{x^{\prime}}) + \sum_{b=1}^B(2 \gamma {w_b^\prime} f(b) ) \right] \right) \right],
\end{equation*}
where $f(b) = \alpha_b (x_b-x_b')$.
This expression can be further simplified by approximating the exponential functions and logarithms by their respective Taylor series for small $\gamma \alpha_b^2, ~ \forall b=1, \ldots, B$. Next, the expectation of the expression over the random vector $\mb{w^\prime}$ can be replaced by an expectation over $\gamma \mb{w^\prime}$, where $\gamma \mb{w^\prime} \sim \mathcal{N}(\mb{0},\frac{\gamma}{2} \mb{I})$ ($\mb{I}$ is the identity matrix). Therefore, we can drop all terms that are proportional to $\mathbb{E}_{\gamma {w_b^\prime}}[\gamma {w_b^\prime}]$ and replace $\mathbb{E}_{\gamma {w_b^\prime}}\left[(\gamma {w_b^\prime})^2\right]$ in all terms proportional to $\mathbb{E}_{\gamma {w_b^\prime}}\left[(\gamma {w_b^\prime})^2\right]$ by $\frac{\gamma}{2}$. Now, after some calculus,we obtain  
\begin{equation}
	I\left(\bs{\alpha}, \gamma, P\right) = \frac{\gamma \sum_b \alpha_b^2 \textrm{Var}(X_b)}{B  \Log(2)} + \sum_{b=1}^B o(\gamma \alpha_b^2),
\label{eq: var}
\end{equation}
where $\textrm{Var}(X_b)$ is the variance of the $b$-th component of the points of constellation $\Omega_x$. The validity of (\ref{eq: var}) for small instantaneous SNR has been verified numerically. As the projection of $\Omega_x$ on either coordinate axis yields the same set of points, this variance is independent of $b$. Hence, for small instantaneous SNR, the mutual information remains constant for the set of fading points where $\sum_{b=1}^{B}\alpha_b^2$ is constant and the outage boundary coincides with a $B$-hypersphere. As a consequence, $\alpha_{\textrm{o}}$ and $\sqrt{B} \alpha_{\textrm{e}}$ are equal and by their definition, it is clear that for low instantaneous SNR, the outage region coincides with the $B$-hypersphere $\sum_{b=1}^{B}\alpha_b^2 = \alpha_{\textrm{o}}^2 = \sqrt{B} \alpha_{\textrm{e}}$. 

\subsection{Proof of Prop. \ref{bound: bpsk with precoding high SNR}}
\label{app: proof prop. precoding high SNR}

Recall that $I\left(\bs{\alpha}, \gamma,  P\right) = \frac{1}{B} I(\mb{T};\mb{Y}|\bs{\alpha}, \gamma)$. For high instantaneous SNR, it was conjectured in \cite{Perez2010, Perez2008} that maximizing (minimizing) the minimal Euclidean distance $d_{\textrm{min}}$ of a constellation maximizes (minimizes) the mutual information on a Gaussian channel. This was based on tight upper and lower bounds on the mutual information, established in \cite{Perez2010, Perez2008} (generalizing \cite{lozano2006}). More specifically, for $d_{\textrm{min}}(\bs{\alpha}) > 0$, 
\begin{equation}
 m - K_1 \frac{e^{-d_{\textrm{min}}^2(\bs{\alpha}) \gamma/4} }{ d^{3}_{\textrm{min}}(\bs{\alpha}) \sqrt{\gamma}} 
\leq I(\mb{T};\mb{Y}|\bs{\alpha}, \gamma) 
\leq m - \frac{e^{-d_{\textrm{min}}^2(\bs{\alpha}) \gamma/4} }{ d_{\textrm{min}}(\bs{\alpha}) \sqrt{\gamma}}\left(K_2 -\frac{K_3}{d^{2}_{\textrm{min}}(\bs{\alpha}) \gamma} \right), \label{eq: mutual information bounds}
\end{equation}
where $d_{\textrm{min}}(\bs{\alpha})$ is the minimal distance of the constellation $\Omega_t$, and where $K_1, K_2$, and $K_3$ are constants. When $d_{\textrm{min}}^2(\bs{\alpha})=0$, the constellation $\Omega_t$ effectively counts less than $2^m$ points and the mutual information $I(\mb{T};\mb{Y}|\bs{\alpha}, \gamma)$ will converge to a value strictly smaller than $m$ for increasing SNR. We now prove that, for $\bs{\alpha}_1, \bs{\alpha}_2$, satisfying $d_{\textrm{min}}^2(\bs{\alpha}_2)>d_{\textrm{min}}^2(\bs{\alpha}_1)$, there exists an SNR-value $\gamma$ above which it is impossible to have 
\begin{equation}
	I(\mb{T};\mb{Y}|\bs{\alpha}_1, \gamma) > I(\mb{T};\mb{Y}|\bs{\alpha}_2, \gamma).	\label{eq: impossible mut info}
\end{equation}	
This is obvious for $d_{\textrm{min}}^2(\bs{\alpha}_1)=0$. Let us check whether (\ref{eq: impossible mut info}) is invalid when $d_{\textrm{min}}^2(\bs{\alpha}_1)>0$ by using contradiction. More specifically, let us assume that (\ref{eq: impossible mut info}) is true. Combining with (\ref{eq: mutual information bounds}), it follows that
\begin{equation*}
m - \frac{e^{-d_{\textrm{min}}^2(\bs{\alpha}_1) \gamma/4 }}{ d_{\textrm{min}}(\bs{\alpha}_1) \sqrt{\gamma}}\left(K_2 -\frac{K_3}{d^{2}_{\textrm{min}}(\bs{\alpha}_1) \gamma} \right) >	m - K_1 \frac{e^{-d_{\textrm{min}}^2(\bs{\alpha}_2) \gamma/4 } }{ d^{3}_{\textrm{min}}(\bs{\alpha}_2) \sqrt{\gamma}},
\end{equation*}
yielding
\begin{equation}
	e^{\frac{\gamma}{4}\underbrace{\left( d_{\textrm{min}}^2(\bs{\alpha}_2) - d_{\textrm{min}}^2(\bs{\alpha}_1)  \right)}_{>0}} < \frac{K_1 d_{\textrm{min}}(\bs{\alpha}_1)}{d_{\textrm{min}}^2(\bs{\alpha}_2) \left( K_2 - K_3/ (d_{\textrm{min}}^2(\bs{\alpha}_1) \gamma)  \right)}. \label{contradiction eq}
\end{equation}
The right-hand side of (\ref{contradiction eq}) is a constant close to $\frac{K_1 d_{\textrm{min}}^2(\bs{\alpha}_1)}{d_{\textrm{min}}^2(\bs{\alpha}_2) K_2}$ for large SNR. Clearly, there exists a large SNR value for which (\ref{contradiction eq}), and thus (\ref{eq: impossible mut info}) is not true, confirming the conjecture made in \cite{Perez2010, Perez2008}.

\textbf{Outer bound:}\\
Consider two points $\mb{x}^{(i)} = (u_1^{(i)}, \ldots, u_{B}^{(i)})$ and $\mb{x}^{(j)} = (u_1^{(j)}, \ldots, u_{B}^{(j)})$ from the constellation $\Omega_x$. The corresponding points $\mb{t}^{(i)}$ and $\mb{t}^{(j)}$ from the faded constellation $\Omega_t$ have a squared Euclidean distance given by 
\begin{equation}
	\left| \mb{t}^{(i)} - \mb{t}^{(j)} \right|^2 = \sum_{b=1}^{B} \alpha_b^2 \left( u_b^{(i)} - u_b^{(j)} \right)^2.
\end{equation}
Let us denote by $b^*$ the value of $b \in \{1, \ldots, B \}$ for which $\left( u_b^{(i)} - u_b^{(j)} \right)^2$ is minimum. For the B-hypersphere $\sum_{b=1}^{B} \alpha_b^2 = \alpha_{\textrm{o}}^2$, we obtain
\begin{equation}
\label{eq: bound high snr}
	\left| \mb{t}^{(i)} - \mb{t}^{(j)} \right|^2 \geq \alpha_{\textrm{o}}^2 \left( u_{b^*}^{(i)} - u_{b^*}^{(j)} \right)^2,
\end{equation}
where equality in Eq. (\ref{eq: bound high snr}) is achieved when $\alpha_{b^*} = \alpha_{\textrm{o}}$. When $b^*$ is not unique, equality in Eq. (\ref{eq: bound high snr}) holds when $\alpha_{b^*} = \alpha_{\textrm{o}}$ holds for any $b^*$ that minimizes $\left( u_b^{(i)} - u_b^{(j)} \right)^2$. Hence, the minimum distance $d_{\textrm{min}}(\bs{\alpha})$ for the constellation $\Omega_t$ is given by
\begin{equation}
	d_{\textrm{min}}(\bs{\alpha}) = \alpha_{\textrm{o}} \min_{\substack{i,j \in \{1, \ldots, M \} \\ i\neq j}} \min_{b\in \{1, \ldots, B \}} \left| u_b^{(i)} - u_b^{(j)} \right|
\end{equation}
As the constellation $\Omega_x$ is such that its projection on either coordinate axis yields the same set of points, the minimum distance $d_{\textrm{min}}(\bs{\alpha})$ on the $B$-hypersphere is achieved for either fading gain equal to $\alpha_{\textrm{o}}$ (and the remaining $B-1$ fading gains equal to $0$). Hence, according to relation between minimal distance and mutual information, $I\left(\bs{\alpha}, \gamma,  P\right)$ with $\bs{\alpha}$ on a $B$-hypersphere is minimized on the coordinate axes at high SNR. Taking $\alpha_{\textrm{o}}$ so that $I\left(\bs{\alpha}, \gamma,  P\right) = R$ on the coordinate axes yields $I\left(\bs{\alpha}, \gamma,  P\right) \geq R$ on the other points of the $B$-hypersphere.

\textbf{Inner bound:}\\
We must prove that the minimal distance of $\Omega_t$ is maximized in $\bs{\alpha}_{\textrm{e}}$, or
\begin{equation}
	\bs{\alpha}_{\textrm{e}} = \argmax{\substack{\bs{\alpha} \\ |\bs{\alpha}|^2 = B \alpha_{\textrm{e}}^2}} ~ \left( \min_{i,j} \sum_{b=1}^{B} \alpha_b^2 \left( u_b^{(i)} - u_b^{(j)} \right)^2 \right)
\end{equation}
For $\bs{\alpha} = \bs{\alpha}_{\textrm{e}}$, we obtain $\sum_{b=1}^{B} \alpha_b^2 \left( u_b^{(i)} - u_b^{(j)} \right)^2 = \alpha_{\textrm{e}}^2 \sum_{b=1}^{B} \left( u_b^{(i)} - u_b^{(j)} \right)^2$, which is minimized to $\alpha_{\textrm{e}}^2 d^2$ when $\mb{x}^{(i)}$ and $\mb{x}^{(j)}$ are points at minimum distance $d$ in the constellation $\Omega_x$. For any $\bs{\alpha}$ with $|\bs{\alpha}|^2 = B \alpha_{\textrm{e}}^2$, $\bs{\alpha} \neq \bs{\alpha}_{\textrm{e}}$, we will prove that $\min_{i,j} \sum_{b=1}^{B} \alpha_b^2 \left( u_b^{(i)} - u_b^{(j)} \right)^2 \leq \alpha_{\textrm{e}}^2 d^2$, which is what is claimed. We first restrict our attention to the case $B>2$. We consider $B$ pairs $\left\{ \left( \mb{x}^{(n)}, \mb{x'}^{(n)} \right), n=0, \ldots, B-1 \right\}$ of points in $\Omega_x$ and their corresponding points $\left\{ \left( \mb{t}^{(n)}, \mb{t'}^{(n)} \right), n=0, \ldots, B-1 \right\}$ in $\Omega_t$, such that $\mb{x}^{(n)} = T^n \mb{x}$ and $\mb{x'}^{(n)} = T^n \mb{x'}$, with $T$ denoting the matrix operator that causes a downward cyclic shift. We select $\mb{x}$ and $\mb{x'}$ such that $|\mb{x}-\mb{x'}| = d$; obviously, $|\mb{x}^{(n)}-\mb{x'}^{(n)}| = d$ for $n=0, \ldots, B-1$. Because $\sum_{b=1}^{B} \left( u_b^{(n)} - u_b^{'(n)} \right)^2 = \sum_{n=0}^{B-1} \left( u_b^{(n)} - u_b^{'(n)} \right)^2$, we obtain that, for any $\bs{\alpha}$ with $|\bs{\alpha}|^2 = B \alpha_{\textrm{e}}^2$, the average distance $\mathbb{E}_n |\mb{t}^{(n)}-\mb{t'}^{(n)}|^2$ is
\begin{equation}
	\frac{1}{N} \sum_{n=0}^{B-1} \left(  \sum_{b=1}^{B} \alpha_b^2 \left( u_b^{(n)} - u_b^{'(n)} \right)^2 \right) = \frac{1}{N} \sum_{b=1}^{B} \alpha_b^2 \left(  \underbrace{\sum_{n=0}^{B-1} \left( u_b^{(n)} - u_b^{'(n)} \right)^2}_{d^2} \right) = d^2 \alpha_{\textrm{e}}^2.
\end{equation}
Hence, $\min_{n} \sum_{b=1}^{B} \alpha_b^2 \left( u_b^{(n)} - u_b^{'(n)} \right)^2 \leq d^2 \alpha_{\textrm{e}}^2$. For $B=2$, the proof has to be slightly modified: we consider 4 pairs $\left\{ \left( \mb{x}^{(n)}, \mb{x'}^{(n)} \right), n=0, \ldots, B-1 \right\}$ of points in $\Omega_x$, such that $\mb{x}^{(n)} = T^n \mb{x}$ and $\mb{x'}^{(n)} = T^n \mb{x'}$, with $T$ denoting the matrix operator representing a counter clockwise rotation of $\pi/2$. We select $\mb{x}$ and $\mb{x'}$ such that $|\mb{x}-\mb{x'}| = d$; obviously, $|\mb{x}^{(n)}-\mb{x'}^{(n)}| = d$ for $n=0, \ldots, 3$. For any $\bs{\alpha}$ with $|\bs{\alpha}|^2 = 2 \alpha_{\textrm{e}}^2$, we have, considering that $T^2 = -I$,
\begin{equation}
	\frac{1}{4} \sum_{n=0}^{3} \left(  \sum_{b=0}^{1} \alpha_b^2 \left( u_b^{(n)} - u_b^{'(n)} \right)^2 \right) = \frac{1}{4} \sum_{b=0}^{1} \alpha_b^2 \left(  \underbrace{\sum_{n=0}^{3} \left( u_b^{(n)} - u_b^{'(n)} \right)^2}_{2 d^2} \right) = d^2 \alpha_{\textrm{e}}^2.
\end{equation}
The remainder of the proof follows the same lines as for $B>2$.  

\subsection{Proof of full diversity} 
\label{app: prop. full diversity} 

Let us find an expression for $\alpha_{\textrm{o}}$. Therefore, we consider that all fading gains are zero, but one, so that the rate achieved by $\Omega_t$ (which is now a scaling of $\mathcal{S}_{\textrm{p}}$) over that single slot must be at least $R B$, or $I_{\mathcal{S}_{\textrm{p}}}(\alpha_{\textrm{o}}^2 \gamma, P) = B R$, so that
\vspace{-0.1cm}
\begin{equation}
	\label{eq: alpha_bo}
 \alpha_{\textrm{o}}^2 = \frac{I_{\mathcal{S}_{\textrm{p}}}^{-1}(B R, P)}{\gamma},
\end{equation}
Now, as the outage region $V_{\textrm{o}}$ is outer bounded by the hypersphere with radius $\alpha_{\textrm{o}}$, we have
\begin{equation}
\label{eq: pout alphabo}
	P_{\textrm{out}}(\gamma, P, R)  \leq \Prob(\alpha_1^2 + \alpha_2^2 + \ldots + \alpha_B^2 < \alpha_{\textrm{o}}^2).
\end{equation}
Because the fading gains $\{ \alpha_b$, $b=1, \ldots, B \}$ are i.i.d. Rayleigh, the cumulative distribution function of the chi-square distribution with $2B$ degrees of freedom is \cite{proakis2000dc} 
\[\Prob(\alpha_1^2 + \alpha_2^2 + \ldots + \alpha_B^2 < x) = 1 - e^{-x} \sum_{k=0}^{B-1} \frac{x^k}{k!} .\]
The right hand side can be manipulated as follows:
\begin{equation*}
	1 - e^{-x} \sum_{k=0}^{B-1} \frac{x^k}{k!} = \frac{e^x - \sum_{k=0}^{B-1}\frac{x^k}{k!} }{e^x} 
	= \frac{\sum_{k=B}^{\infty}\frac{x^k}{k!} }{\sum_{m=0}^{\infty}\frac{x^m}{m!}} 
	= \frac{x^B \sum_{k=0}^{\infty}\frac{x^k}{(B+k)!} }{\sum_{m=0}^{\infty}\frac{x^m}{m!}} 
	\leq x^B,
\end{equation*}
so that from Eqs. (\ref{eq: alpha_bo}) and (\ref{eq: pout alphabo})
\begin{equation}
	P_{\textrm{out}}(\gamma, P, R)  \leq \alpha_{\textrm{o}}^{2B} \propto \frac{1}{\gamma^B}.
\end{equation}
Hence, with increasing $\gamma$ the outage probability goes to zero not slower than $\gamma^{-B}$, so the diversity order is at least $B$. As there are only $B$ i.i.d. fading gains involved in the transmission of a codeword, the diversity order is exactly $B$.

\end{document}